\documentclass[a4paper,UKenglish,pdfa,cleveref,autoref,thm-restate,numberwithinsect]{lipics-v2021}

\usepackage{rotating}
\usepackage{adjustbox}
\usepackage{bussproofs}

\usepackage{tikz}
\usetikzlibrary{calc,intersections,decorations.pathmorphing,decorations.markings,patterns,positioning,arrows.meta,shapes}
\tikzset{double_border/.style={draw, double, double distance=1pt,outer sep=1.2pt}}

\title{On the Interplay of Cube Learning and Dependency Schemes in QCDCL Proof Systems} 

\titlerunning{Cube Learning and Dependency Schemes in QCDCL} 

\author{Abhimanyu Choudhury}{The Institute of Mathematical Sciences, Chennai, India \and Homi Bhabha National Institute, Training School Complex, Anushaktinagar, Mumbai, India}{abhimanyuc@imsc.res.in}{https://orcid.org/0009-0003-7659-5995}{} 

\author{Meena Mahajan}{The Institute of Mathematical Sciences, Chennai, India \and Homi Bhabha National Institute, Training School Complex, Anushaktinagar, Mumbai, India}{meena@imsc.res.in}{https://orcid.org/0000-0002-9116-4398}{Partially supported by J C Bose National Fellowship.}

\authorrunning{A.~Choudhury and M.~Mahajan} 

\Copyright{Abhimanyu Choudhury and Meena Mahajan} 



\category{} 

\relatedversion{} 

\nolinenumbers 

\hideLIPIcs  

\EventEditors{John Q. Open and Joan R. Access}
\EventNoEds{2}
\EventLongTitle{42nd Conference on Very Important Topics (CVIT 2016)}
\EventShortTitle{CVIT 2016}
\EventAcronym{CVIT}
\EventYear{2016}
\EventDate{December 24--27, 2016}
\EventLocation{Little Whinging, United Kingdom}
\EventLogo{}
\SeriesVolume{42}
\ArticleNo{23}

\newcommand{\calQ}{\mathcal{Q}} 
\newcommand{\D}{\mathtt{D}} 
\newcommand{\CAxiom}{\mbox{$\mathtt{Axiom}$}} 
\newcommand{\TAxiom}{\mbox{$\mathtt{Axiom}$}} 
\newcommand{\Res}{\mbox{$\mathtt{Res}$}} 
\newcommand{\TermRes}{\mbox{$\mathtt{TermRes}$}} 
\newcommand{\PS}{\mathtt{P}} 

\newcommand{\var}{\mathtt{var}}

\newcommand{\rrs}{\mathtt{rrs}}
\newcommand{\respath}{\mathtt{res}}
\newcommand{\trv}{\mathtt{trv}}
\newcommand{\std}{\mathtt{std}}
\newcommand{\lo}{\mathtt{LEV\textrm{-}ORD}}

\newcommand{\uniany}{\mathtt{UNI\textrm{-}ANY}}

\newcommand{\aro}{\mathtt{ASS\textrm{-}R\textrm{-}ORD}}
\newcommand{\ao}{\mathtt{ANY\textrm{-}ORD}}
\newcommand{\resolve}{\mathtt{res}}
\newcommand{\ante}{\mathtt{ante}}
\newcommand{\reduce}{\mathtt{red}}
\newcommand{\cube}{\mathtt{cube}}
\newcommand{\reduceD}{\mathtt{red}\mbox{-}\D}
\newcommand{\reduceDrrs}{\mathtt{red}\mbox{-}\Drrs}
\newcommand{\reduceDstd}{\mathtt{red}\mbox{-}\Dstd}
\newcommand{\red}{\mathtt{RED}}

\newcommand{\Dres}{\D^{\respath}}
\newcommand{\Drrs}{\D^{\rrs}}
\newcommand{\Dtrv}{\D^{\trv}}
\newcommand{\Dstd}{\D^{\std}}

\newcommand{\QCDCL}{\mbox{$\mathtt{QCDCL}$}} 
\newcommand{\QCDCLlored}{\mbox{$\QCDCL^{\lo}_{\red}$}}

\newcommand{\QCDCLcube}{\mbox{$\QCDCL^{\cube}$}}

\newcommand{\QCDCLDrrs}{\mbox{$\QCDCL(\Drrs)$}}

\newcommand{\DrrsQCDCLcube}{\mbox{$\Drrs + \QCDCL^{\cube}$}}

\newcommand{\DrrsQCDCLcubeDrrs}{\mbox{$\Drrs+ \QCDCL^{\cube}(\Drrs)$}}

\newcommand{\QRes}{\mbox{$\mathtt{Q}$-$\Res$}} 
\newcommand{\QDRes}{\mbox{$\mathtt{Q}(\D)$-$\Res$}}  
\newcommand{\QDrrsRes}{\mbox{$\mathtt{Q}(\Drrs)$-$\Res$}} 
\newcommand{\QDstdRes}{\mbox{$\mathtt{Q}(\Dstd)$-$\Res$}} 
\newcommand{\QURes}{\mbox{$\mathtt{QU}$-$\Res$}} 
\newcommand{\LDQRes}{\mbox{$\mathtt{LDQ}$-$\Res$}} 
\newcommand{\LDQTermRes}{\mbox{$\mathtt{LDQ}$-$\TermRes$}} 
\newcommand{\LDQDRes}{\mbox{$\mathtt{LDQ}(\D)$-$\Res$}}
\newcommand{\LDQDTermRes}{\mbox{$\mathtt{LDQ}(\D)$-$\TermRes$}}

\newcommand{\PHP}{\mbox{$\mathtt{PHP}$}}
\newcommand{\Equality}{\mbox{$\mathtt{Equality}$}}
\newcommand{\Trapdoor}{\mbox{$\mathtt{Trapdoor}$}}

\newcommand{\TwinEquality}{\mbox{$\mathtt{TwinEq}$}}
\newcommand{\CR}{\mbox{$\mathtt{\CR}$}}

\newcommand{\DepTrap}{\mbox{$\mathtt{Dep\textrm{-}Trap}$}}

\newcommand{\TwoPHPandCT}{\mbox{$\mathtt{TwoPHPandCT}$}}

\newcommand{\PreDepTrap}{\mbox{$\mathtt{PreDepTrap}$}}

\newcommand{\PropDepTrap}{\mbox{$\mathtt{PropDep\textrm{-}Trap}$}}

\newcommand{\DoubleLongEq}{\mbox{$\mathtt{DoubleLongEq}$}}

\newcommand{\PreRRSTrap}{\mbox{$\mathtt{PreRRSTrapdoor}$}}

\newcommand{\StdDepTrap}{\mbox{$\mathtt{StdDepTrap}$}}

\newcommand{\redPreRRSTrap}{\mbox{$\reduceDrrs(\PreRRSTrap)$}}

\newcommand{\T}{\mathcal{T}}

\newcommand{\Ord}{\mathtt{ORD}}
\newcommand{\Dord}[1]{{#1}\textrm{-}\mathtt{ORD}}
\newcommand{\QCDCLDep}[3]{\mbox{$\QCDCL^{#1}(#2,#3)$}}
\newcommand{\ClausePolicy}{\mbox{$\mathtt{ClausePol}$}}
\newcommand{\CubePolicy}{\mbox{$\mathtt{CubePol}$}}
\newcommand{\NoCube}{\mbox{$\mathtt{No}$-$\mathtt{Cube}$}}
\newcommand{\CubeLD}{\mbox{$\mathtt{Cube}$-$\mathtt{LD}$}}
\newcommand{\CubeD}[1]{\mbox{$\mathtt{Cube}$-$#1$}}
\newcommand{\QTermRes}{\mbox{$\mathtt{Q}$-$\TermRes$}}
\newcommand{\QDTermRes}{\mbox{$\mathtt{Q}(\D)$-$\TermRes$}}
\newcommand{\LDRD}[1]{\mathtt{LD}\Res(#1)} 
\newcommand{\LDTRD}[1]{\mathtt{LDT}\Res(#1)} 

\begin{document}
	
	\maketitle
	
\begin{abstract}
Quantified Conflict Driven Clause Leaning (QCDCL) is one of the main approaches to solving Quantified Boolean Formulas (QBF). Cube-learning  is employed in this approach to ensure that true formulas can be verified. 
Dependency Schemes help to detect spurious dependencies that are implied by the variable ordering in the quantifier prefix of QBFs but are not essential for constructing (counter)models. This detection can provably shorten refutations in specific proof systems, and is expected to speed up runs of QBF solvers.

The simplest underlying proof system [BeyersdorffB\"ohm-LMCS2023],
formalises the reasoning in the QCDCL approach on false formulas, when neither cube-learning nor dependency schemes is used.
The work of [B\"ohmPeitlBeyersdorff-AI2024] further incorporates cube-learning. 
The work of [ChoudhuryMahajan-JAR2024] incorporates a limited use of  dependency schemes, but without cube-learning. 

In this work, proof systems underlying the reasoning of QCDCL solvers
which use cube learning, and which use dependency schemes at all
stages, are formalised. Sufficient conditions for soundness and
completeness are presented, and it is shown that using the standard
and reflexive resolution path dependency schemes ($\Dstd$ and $\Drrs$)
to relax the decision order provably shortens refutations.

When the decisions are restricted to follow quantification order, but
dependency schemes are used in propagation and learning, in
conjunction with cube-learning, the resulting proof systems using the
dependency schemes $\Dstd$ and $\Drrs$ are investigated in detail and
their relative strengths are analysed.
\end{abstract}
	
        
	\section{Introduction} 
	\label{sec:Intro}

Despite the NP-hardness of the satisfiability problem, in the
last three decades SAT solvers have been phenomenally
successful in solving instances of humongous size, and have
become the go-to tool in many practical industrial
applications (see e.g.\ \cite{Var14,sathandbookcdcl}).
This success has spurred ambitious programs to develop solvers
for computationally even more hard problems. In particular,
the PSPACE-complete problem of determining the truth of
Quantified Boolean Formulas QBFs has many more applications
(see e.g.\ \cite{ShuklaBPS19}), and over the last
twenty years QBF solvers have rapidly approached the state of
industrial applicability.

The paradigm that revolutionized SAT solving is Conflict Driven Clause
Learning CDCL (\cite{SilvaSakallah-ICCAD96}), and this is also one of
the principal approaches (but not the only one) in QBF solving. The
CDCL technique was lifted to QBFs in the form of QCDCL
(\cite{ZhangMalik-ICCAD02}, see also \cite{GMN-HandbookSAT09}; in \cite{Lonsing-Thesis12}, the term QDPLL is used), and
implemented in state-of-the-art solvers DepQBF
\cite{LonsingBiere-JSAT10,LonsingE17} and Qute \cite{PeitlSS-JAIR19}
with further augmentations to enhance performance.

For both SAT and QBF, solving techniques are intricately connected
with proof systems. The runtime trace of a solver on a formula can be
thought of as a proof of the final outcome (sat/unsat,
true/false). Proof systems abstract out the reasoning employed in the
solvers, and allow representing these traces-as-proofs as formal
proofs. The CDCL paradigm in SAT solvers corresponds to resolution, a
very well-studied proof system. There are multiple
ways in which resolution can be lifted to QBFs, see
\cite{M4CQBF} for an overview.  As
shown in \cite{BeameKS-JAIR04}, resolution proofs can be efficiently
extracted from traces of CDCL-based SAT solvers. For QBFs, QCDCL
traces yield proofs in the proof system long-distance Q-resolution
\LDQRes~ \cite{ZhangMalik-ICCAD02,BalabanovJiang-FMSD12}. However, the
converse direction, going from resolution proofs to CDCL runs,
famously shown for SAT in
\cite{PipatsrisawatDarwiche-AI11,AsteriasFT-JAIR11}, seemingly breaks
down for QBF and QCDCL as currently implemented; the reasoning
employed in basic QCDCL solvers was abstracted in
\cite{BeyersdorffBohm-LMCS23,BohmPB-AI24} as the proof systems
\QCDCL\ and \QCDCLcube, and shown to be exponentially weaker than \LDQRes.

The proof system \QCDCL\ is a refutational proof system; it was
formulated in \cite{BeyersdorffBohm-LMCS23} to explain the reasoning of basic QCDCL-style algorithms on
false QBFs. The proof system
\QCDCLcube, defined in \cite{BohmPB-AI24}, incorporates cube-learning as
well, and can thus certify both falsity and truth of
QBFs. Intriguingly, it was shown in \cite{BohmPB-AI24} that even
for false QBFs, where cube-learning is not necessary for completeness,
it can still significantly shorten refutations. 
Very recently, it was
shown in \cite{BeyersdorffBohmMahajan-AAAI24} that even when short
refutations are actually found, it may take an exponentially long time to
find them.
Many other variants (different  policies for decision order, propagations, reductions) have been studied extensively in \cite{BohmB-JAIR24}. 

One heuristic that has been used in some QCDCL solvers is the use of
dependency schemes. These schemes involve performing some basic
analysis of the formula structure and identifying spurious
dependencies amongst variables, dependencies that are implied by the
quantification order of variables but are not necessary for
constructing (counter)models; see
for instance \cite{SlivovskySzeider-TCS16}. Eliminating  such dependencies would
transform a QBF to a Dependency QBF, DQBF, for which the computational
problem of deciding truth/falsity is even harder; it is NEXP-complete
(\cite{AzharPR-CMA01,BlinkhornPS-SAT21}). However, retaining the
formulas as a QBF, and using information about spurious
dependencies while propagating and learning, is still a
feasible approach, that has been implemented in the solver DepQBF
\cite{LonsingBiere-JSAT10,LonsingE17} using the {\em
	standard dependency scheme} $\Dstd$. In resolution-based QBF proof
systems, employing reduction rules based on the {\em reflexive
	resolution path dependency-scheme} $\Drrs$, is known to exponentially
shorten refutations (\cite{BlinkhornBeyersdorff-SAT17}), and the
expectation is that a similar advantage also manifests in QCDCL
solvers.

This work makes progress towards formally understanding the
strengths/limitations of using the dependency-scheme heuristic.  The first
steps in this direction were initiated in a recent work in
\cite{ChoudhuryMahajan-JAR24}. It considered the simplest setting, in
which the \QCDCL\ proof system uses the $\lo$ decision policy
(deciding variables according to the quantification order), and does
not learn cubes. 
A
dependency scheme $\D$ is used in propagation by, and learning of,
clauses. Additionally, a dependency scheme $\D'$ may be used to
preprocess the formula, reducing all clauses according to $\D'$ 
before beginning the QCDCL trails.  In this setting,
when $\D$ and $\D'$ are ``normal'' schemes (as
defined in \cite{PeitlSS-JAR19}), the resulting proof systems were shown to be 
sound and refutationally complete. 
In the same  setting, the four systems arising from using $\Drrs$ in preprocessing, in propagation/learning, in both, and in neither, 
were shown to be incomparable with each other.
In the 
underlying proof system \LDQRes, using dependency information can
never lengthen proofs. The handicap in QCDCL arises because QCDCL
algorithms also need to search for the proof.

In this work, we consider more general settings. Our contributions
are as follows:

\noindent {\bf Formalising intensive use of dependency schemes in QCDCL:}
We formalise the definitions of $\QCDCL$ and $\QCDCLcube$ proof
systems that use dependency schemes more intensively: in the decision
policy, which determines which variables can be "decided" at a particular stage, as well as in propagation and learning, with and without cube
learning. Using a dependency scheme $\D_1$ in the decision policy
means that a variable can be decided if all variables on which it
depends, according to $\D_1$, are already assigned; this is the policy
$\Dord{\D_1}$.  Using a dependency scheme $\D_2$ in propagation and
learning means that reductions enabled by $\D_2$ are performed
whenever possible.  For two dependency schemes $\D_1$ and $\D_2$
(which may be the same) we consider $\QCDCLcube$ proof systems that
use $\D_1$ in the decision policy and $\D_2$ in propagation through and
learning of clauses. We consider three scenarios with respect to cube-learning:
(1)~cube-learning is switched off completely; (2)~cube propagation and
learning is done without using any dependency schemes; or (3)~cube
propagation and learning use the scheme $\D_2$ but disallow
long-distance term-resolution. The reason for this difference between
clause and cube learning is that long-distance term resolution is not
(yet?) known to be sound if used in conjunction with dependency
schemes.
We show that for normal $\D_1$, $\D_2$, the resulting systems are
sound and refutationally complete; \cref{thm:depord-sound-complete}.

\noindent {\bf Provable advantage of $\Dord{\D}$:}
We show that other parameters remaining the same, using either
$\Drrs$ or $\Dstd$ as $\D_1$ is strictly better than using $\lo$; \cref{thm:depord-levord}.


\noindent {\bf Using cube-learning, and dependency schemes only in propagation/learning:}
When the decision policy is restricted to $\lo$, we generalise the results from
\cite{ChoudhuryMahajan-JAR24} to settings with cube-learning switched on, and also to settings where $\Dstd$ rather than $\Drrs$ is used. Specifically,
we show that
\begin{enumerate}
\item Using $\Dstd$ in pre-processing is useless; \cref{prop:DtrvDstd-nopre}.
  
	\item Switching on cube learning provably adds power
	  even in the presence of $\Dstd$ or $\Drrs$;
          \cref{thm:QCDCLcubeD-stronger}.
	
	\item Adding $\Drrs$ in various non-trivial ways to
	\QCDCLcube\ results in proof systems that are not only pairwise
	incomparable, \cref{thm:QCDCLcubeDrrs-incomparable}, but are also incomparable with both
	\QCDCLcube\ and \QCDCL; \cref{thm:QCDCLcubeDrrs-QCDCLcube,thm:QCDCLcubeDrrs-DrrsQCDCL}. 
	
	\item Adding $\Dstd$ to \QCDCL\ is orthogonal to switching on
	cube learning; \cref{thm:DstdCube-incomp}. In certain cases, adding $\Dstd$
	to both \QCDCL\ and \QCDCLcube\ offers a provable advantage.
	
	\item Although $\Drrs$ strictly refines $\Dstd$, in
	the context of \QCDCL\ and \QCDCLcube, adding these schemes
	gives rise to incomparable systems; \cref{thm:DstdDrrs-incomp}. Thus, the  $\lo$ policy can negate
	potential benefits of the strict refinement.
\end{enumerate}
We use several known bounds on formulas from earlier
works, and also show some new bounds for them. To obtain desired separations, we also introduce three carefully handcrafted new formulas. 
For easy reference, the known and new results (about
previously defined and new formulas) are collated in
\cref{tab:for-sideways}  in \cref{sec:QCDCLcuberesults}. The known simulation order of the proof systems, incorporating prior known results as well as the new results proved here, are summarised in \cref{fig:all-figures}, also in \cref{sec:QCDCLcuberesults}.

\paragraph*{Organisation of the paper:}
In \cref{sec:Prelim}, we give some basic definitions and describe the background about known proof systems and dependency schemes. In \cref{sec:QCDCLDep}, we define the new proof systems, show soundness and completeness, and show that the decision policy $\Dord{\D}$ is strictly more powerful than $\lo$. In \cref{sec:QCDCLcuberesults}, we briefly discuss preprocessing, we define three new QBF families and show various lower and upper bounds for their proof sizes, and we describe the simulation order among various \QCDCL\ systems. We conclude with some pointers for further directions of interest.


	\section{Preliminaries}
	\label{sec:Prelim}
	
	\subsection{Basic Notation}\label{sec:basics}
        We follow notation from 
\cite{BeyersdorffBohm-LMCS23,BohmPB-AI24}; see also \cite{M4CQBF}.
Selected relevant items are included here.

	A literal $\ell$ is a Boolean variable $x$ or its negation $\bar{x}$, and $\var(\ell)$ denotes the associated variable $x$. A clause is a disjunction of literals; a term or a cube is a conjunction of literals. For a clause or cube $C$, $\var(C)$ denotes the set $\{\var(\ell) \mid \ell \in C\}$. A propositional formula $\varphi$ is built from variables using conjunction, disjunction, and negation; it is in conjunctive normal form~(CNF) if it is a conjunction of clauses. For a formula $\varphi$, a variable $x$ in $\varphi$, and a Boolean value $a$, $\varphi|_{x=a}$ refers to the formula obtained by substituting $x=a$ everywhere in $\varphi$.  For a set $S$ of clauses and a literal $\ell$, we use shorthand $\ell \vee S$ to denote the set of clauses $\{\ell\vee C \mid C\in S\}$. The
	empty clause is denoted $\square$ and is unsatisfiable; the empty cube is denoted $\top$ and is always true. A clause (cube) is said to be tautological (contradictory) if for some variable $x$ it contains both $x$ and $\bar{x}$.
	
	The resolution rule can be applied to clauses and to cubes.
        The resolvent of  clauses $A'=A\vee x$ and $B'=B\vee \bar{x}$ is the clause $A\vee B$ denoted as $\resolve(A',B',x)$ or $\resolve(B',A',x)$.	The resolvent of cubes $A'=A\wedge x$ and $B'=B\wedge \bar{x}$ is  the cube $A\wedge B$, also  denoted as $\resolve(A',B',x)$ or $\resolve(B',A',x)$. 
	
	A Quantified Boolean Formula (QBF) in prenex conjunction normal form (PCNF) is a prefix with a list of variables, each quantified either existentially 
	or universally, and a matrix, which is a set (conjunction) of clauses over these variables. That is, it has the form  \[ \Phi = \calQ \vec{x} \cdot \varphi = Q_1 x_1 Q_2 x_2 \ldots Q_n x_n ~~ \varphi(x_1,x_2,\ldots ,x_n)\] where $\varphi$ is a propositional formula in CNF, and each $Q_i$ is in $\{\exists,\forall\}$. We denote by $X_\exists$ ($X_\forall$ respectively) the set of all variables quantified existentially (resp.\ universally).
	
	A QBF is true if for each existentially quantified variable $x_i$, there exists a (Skolem) function $s_i$, depending only on universally quantified variables $x_j$ with $j<i$, such that
	substituting these $s_i$ in $\varphi$ yields a tautology. Similarly, the formula is false if for each universally quantified variable $u_i$, there is a (Herbrand) function $h_i$, depending only on existentially quantified variables $x_j$ with $j<i$, such that substituting $h_i$ in $\varphi$ yields an unsatisfiable formula.

        \subsection{Some QBF proof systems, and the Dependency Scheme heuristic}\label{subsec:priorsystems}
        The propositional proof system Resolution certifies unsatisfiability of a propositional formula by adding clauses obtaining through resolution until the empty clause is added. This can be lifted to QBFs in many ways. One of the simplest ways is to use the resolution rule along  with a universal reduction rule, that allows removing a universal literal $u$ or $\bar{u}$ from a clause if it is not `{\em blocked}; that is, the clause has no existential literals quantified after $u$ in the prefix. This gives rise to the system \QURes; its restriction where resolution is allowed only on existential pivots is the system \QRes. The long-distance resolution rule generalises resolution by permitting seemingly useless universal tautological clauses under certain side-conditions, and gives rise to the system \LDQRes\ that generalises \QRes. Informally, in this system, a resolution on $x$ is permitted even if the resolvent ends up having $u$ and $\bar{u}$ for some universal variable $u$, provided $u$ is quantified to the right of $x$. The presence of $u$ and $\bar{u}$ together, often referred to as a {\em merged literal} $u^*$, is to be interpreted not as a tautology but as a place-holder for a partial strategy for $u$ depending on the value of the pivot $x$.

        In direct analogy to \QRes\ and \LDQRes\  are the proof systems \QTermRes\  and \LDQTermRes\ for certifying truth. Here the resolution is performed on terms, or cubes, with universal pivots, and existential literals can be reduced from a term if not blocked by universal literals quantified after them. The goal is to derive the empty term. A key point of difference is at the start; since the QBF is in PCNF, there are no terms to begin with. The Axiom  rule in these systems permits starting with any term that satisfies the matrix.

        For formal definitions of these proof systems, see for instance, Figure~2 in  \cite{BeyersdorffCJ-ToCT19} (for \QRes,\QURes, \LDQRes), Figure~2, in \cite{SlivovskySzeider-TCS16} (for \QTermRes).  To help readability, we also include the definitions of the rules in the appendix. 
        
        The notion of blocking, used in the reduction rules, stems from the understanding that if variables $x,y$ are quantified differently with $y$ quantified after $x$, then the value of $y$ in a (counter)model may non-trivially {\em depend} on $x$. If there is no literal blocking $x$, then the satisfaction of the clause (falsification of cube) should not rely on the unblocked universal (resp.\ existential) $x$. The dependency scheme heuristic refines this further. If a syntactic examination of the clause-variable structure can reveal that $y$ does not really depend on $x$, {\em even though it is quantified later}, then $x$ can be reduced even in the presence of $y$. This can drive the process towards the empty clause/term faster. Dependency schemes do precisely this. They identify pairs $(x,y)$ where $x$ and $y$ are quantified differently and where $y$ can be safely assumed to be independent of $x$. (Actually, they list pairs where $y$ may depend on $x$; the other pairs can be assumed to be independent.)
        The trivial dependency scheme associates with each QBF $\Phi$ the dependency set $\Dtrv(\Phi) = \{(x,y) \mid \textrm{$y$ is quantifed after, and differently from, $x$} \}$. Other schemes can associate subsets of this set. The schemes relevant to this paper are the standard scheme $\Dstd$ (Def~7 in \cite{SamerSzeider-JAR09} and Def~9 in \cite{SlivovskySzeider-TCS16}), and the reflexive resolution scheme $\Drrs$ (Defs~3,4,6 in \cite{SlivovskySzeider-TCS16}); these definitions are reproduced below. For every $\Phi$, $\Dtrv(\Phi) \supseteq \Dstd(\Phi) \supseteq \Drrs(\Phi)$. 
        \begin{definition}[Standard Dependency Scheme, {\cite[Def 9]{SlivovskySzeider-TCS16}}] 
	\label{def:std}
	For a PCNF QBF $\Phi=\calQ \vec{x} \cdot \varphi$, the pair $(x,y)$ is in $\Dstd(\Phi)$ if and only if $(x,y) \in \Dtrv(\Phi)$ and there exists a sequence of clauses $C_1, \cdots, C_n \in \varphi$ and a sequence of existential literals $\ell_1,\cdots,\ell_{n-1}$ such that:
	\begin{itemize}
		\item $x \in C_1$ and $y \in C_n$, and
		\item for each $i\in[n-1]$, $(x,\var(\ell_i))
		\in \Dtrv(\Phi)$, $\var(\ell_i) \in \var(C_i)$, and $\var(\ell_i) \in \var(C_{i+1})$.
	\end{itemize}
\end{definition}

\begin{definition}[Reflexive Resolution Path Dependency Scheme, {\cite[Defs 3,4,6]{SlivovskySzeider-TCS16}}]        
	\label{def:rrs}
	
	Fix any PCNF QBF $\Phi=\calQ \vec{x} \cdot \varphi$.
	
	An ordered pair of literals $\ell_1,\ell_{2k}$ is  connected (via a resolution path) if there is a sequence of clauses $C_1, \cdots, C_k \in \varphi$ and a sequence of existential literals $\ell_2,\cdots,\ell_{2k-1}$ such that:
	\begin{itemize}
		\item $\ell_1 \in C_1$ and $\ell_{2k} \in C_k$,
		\item For each $i\in [k-1]$, $\ell_{2i} = \neg \ell_{2i+1}$.
		\item For each  $i \in [k-1]$, $(\var(\ell_1),\var(\ell_{2i})) \in \Dtrv$. 
		\item For each $i\in [k]$, $\var(\ell_{2i-1}) \neq \var(\ell_{2i})$. 
		\item For each $i\in [k]$, $\ell_{2i-1}, \ell_{2i} \in C_i$.
	\end{itemize}
	
	An ordered pair of variables $(x,y)$ is a resolution-path dependency pair if
	both $(x,y)$ and $(\neg x, \neg y)$ are connected, or if 
	both $(x,\neg y)$ and $(\neg x, y)$ are connected.
	\[\Drrs(\Phi) = \{ (x,y) \mid (x,y) \in \Dtrv ; \textrm{$(x,y)$ is
		a resolution-path dependency pair} \}.\]
\end{definition}

Roughly, $(x,y) \in \Dstd(\Phi)$ if there is a sequence of clauses with the first containing $x$ or $\bar{x}$, the last containing $y$ or $\bar{y}$, and each pair of consecutive clauses containing an existential variable quantified to the right of $x$.  For $\Drrs$ scheme, $(x,y) \in \Drrs(\Phi)$ if $(x,y)$ and $(\bar{x},\bar{y})$ or $(x,\bar{y})$ and $(\bar{x},y)$ are connected by a sequence of clauses, each pair of consecutive clauses containing an existential variable in opposite polarities  quantified to the right of $x$.

        When  a dependency scheme $D$ is incorporated into any of the preceding proof systems, the reduction rule becomes more generally applicable, and the side-conditions concerning merged literals also become more permissive. This gives rise to the proof systems \QDRes, \LDQDRes, \QDTermRes, \LDQDTermRes. See,  for instance, Figure~3 and Section~3.3 in \cite{SlivovskySzeider-TCS16} (for \QTermRes\ and $\D$-reductions), and Figure~1 in \cite{PeitlSS-JAR19} (for \LDQDRes). For a clause and a dependency scheme $\D$, we denote by $\reduceD(C)$ the clause  obtained by removing all unblocked universal literals from $C$. Similarly, for a cube $C$, $\reduceD_\exists(C)$ denotes the cube obtained  by removing all unblocked existential literals from $C$.  We denote by $\reduceD(\Phi)$ the QBF $\Psi$ obtained by replacing each clause $C$ in the matrix of $\Phi$ with the clause $\reduceD(C)$.  When $\D=\Dtrv$, we use the notation $\reduce(C)$ and $\reduce(\Phi)$.

        An  important subclass of dependency schemes are the  so-called {\em normal} dependency schemes, which have the property of being "monotone" and "simple". See Def.~7 in \cite{PeitlSS-JAR19} for the precise definition. 
        (Though we will not need the precise definitions, for completeness, we include the definition of normal  dependency schemes in the appendix.)
        The dependency schemes $\Dtrv$, $\Dstd$, $\Drrs$ are all normal. This class of schemes  is of interest to us because it is known that for normal dependency schemes, $\QDRes$ and \LDQDRes\ are sound and refutationally complete \cite{PeitlSS-JAR19}. The system $\QDTermRes$ is known to be sound and complete on true formulas for $\D \in \{\Dtrv, \Drrs\}$ (in \cite{SlivovskySzeider-TCS16}, soundness is shown for a stronger dependency scheme, $\Dres$, implying soundness for $\Drrs$ as well). However the soundness of $\LDQDTermRes$ is not known for  dependency schemes other than $\Dtrv$.

    	We say  proof system $\PS_1$ simulates a proof system $\PS_2$ if some computable function transforms proofs in $\PS_2$ into proofs in $\PS_1$ with at most polynomial blow-up in proof size. If this function is also computable in polynomial time (in the given proof size), we say that $\PS_1$ $p$-simulates $\PS_2$. Two systems are said to be incomparable if neither simulates the other. 
        
By definition \LDQDRes\ $p$-simulates \LDQRes\ and \QDRes, both of which $p$-simulate \QRes. Similarly, both \LDQTermRes\ and \QDTermRes\ $p$-simulate \QTermRes.  It is also  known that \QDRes\ is exponentially stronger than \QRes\ for $\D=\Drrs$; \cite{BlinkhornBeyersdorff-SAT17,BlinkhornBeyersdorff-IJCAI18}.

	\subsection{The proof systems \texorpdfstring{\QCDCL\  with and without cube learning}{QCDCL and QCDCLcube}}
	The proof system for \QCDCL, as defined in \cite{BeyersdorffBohm-LMCS23}, formalizes reasoning in QCDCL algorithms operating on false formulas without cube learning. These algorithms construct {\em trails} or partial assignments in a specific way -- {\em decide} values of variables according to some policy, {\em propagate} values of existential variables that appear in clauses which become unit after restriction by the trail so far and by universal reduction (call such a clause the {\em antecedent} of the propagated literal) -- trying to satisfy the matrix. If a conflict is reached,  then the trail is inconsistent with any Skolem function. {\em Conflict analysis} is performed, and a new clause is {\em learnt} and added to the matrix. If the empty clause is learned, the formula is deemed false. The corresponding refutation in the proof system \QCDCL\ consists of the sequence of constructed trails, and for each trail the sequence of long-distance resolution  steps performed in conflict analysis to learn a clause. The full definition can be found in \cite{BeyersdorffBohm-LMCS23} (Def.\ 3.5).
	
	The above formulation of the \QCDCL\ system only considers trails  
	that end in a conflict. Trails ending in a satisfying assignment are ignored. This is enough to ensure refutational completeness -- the ability to prove all false QBFs false. However, from satisfying assignments, solvers can and do learn cubes (or terms), and this is necessary to prove true QBFs true. In \cite{BohmPB-AI24} it was shown that even  while refuting false QBFs, allowing cube  learning from satisfying assignments can be advantageous. This led to the definition of the proof system $\QCDCLcube$, and in \cite{BohmPB-AI24}, it was shown to be strictly stronger than the  standard \QCDCL\ system. 
	The main idea in cube learning systems is to consider satisfying assignments also as conflicts, albeit of a different kind, and to learn cubes from these conflicts. (The algorithm learns that such a trail is inconsistent with any Herbrand function.) Learnt cubes are added disjunctively to the matrix, which thus at intermediate stages is not necessarily in CNF but is the disjunction of a CNF formula and some cubes. With the augmented CNF matrices, cubes that become unit after existential reduction  can now propagate  universal variables in a way that falsifies the cube.  Also, with the  augmented CNF matrices, a trail may end up satisfying a cube rather than the CNF; this too is now a conflict, and conflict analysis involving term-resolution can be performed to learn a new cube.
        For formal details, see Section~3 in \cite{BohmPB-AI24}. 
	
        Three factors affect the construction of a refutation or verification, and are relevant for our generalized definition in the following section:
Three factors affect the construction of the refutation.
\begin{enumerate}
	\item 	The decision policy: how to choose the next variable to branch
		on.  In standard \QCDCL\ as defined in \cite{BeyersdorffBohm-LMCS23}, 
		decisions must
		respect the quantifier prefix level order. (Variables $x,y$ are at
		the same level if they are quantified the same way, and no variable
		with a different quantification appears between them in the prefix
		order.) This policy is called $\lo$. The most unrestricted policy is $\ao$; any variable can be decided at any point.

Other policies such as $\ao$, $\aro$, $\uniany$, are also
possible; see \cite{BeyersdorffBohm-LMCS23, BohmPB-JAR24}.
		
\item The unit propagation policy.  Upon a partial assignment $\alpha$
to some variables, when does a clause $C$ propagate a literal?
In standard \QCDCL\, the Reduction policy (used by most current QCDCL solvers \cite{LonsingBiere-JSAT10,PeitlSS-JAIR19}) is used: a clause $C$
propagates literal $\ell$ if after restricting $C$ by $\alpha$ and
applying all possible universal reductions, only $\ell$ remains.
Also, propagations are made as soon as possible; see the description of natural trails (Def 3.4 in \cite{BeyersdorffBohm-LMCS23}). 
		

\item The set of learnable clauses/cubes. These clauses/cubes explain the conflict at the end of a trail, and are derived using (possibly long-distance) clause/term resolution with propagated literals as pivots.
	\end{enumerate}

	\section{Adding dependency schemes to the \QCDCL\ proof system}
	\label{sec:QCDCLDep}
	
	The work done in \cite{ChoudhuryMahajan-JAR24}  was the first to formalise the addition of dependency schemes to the \QCDCL\ proof system. It was done only for the setting where trails follow level-ordered decisions, and there is no cube learning. Here we relax both these restrictions; we  
allow  dependency schemes to affect the decision policy of the trail, and we allow cube-learning.

\subsection{Defining the $\Dord{\D}$ systems}
\label{subsec:depord}
Dependency schemes can be used in QCDCL algorithms in many ways: \\
(1)~in specifying the decision order, (2)~in specifying how reduction, propagation, and learning of clauses are performed, and (3)~in specifying how these are  performed for cubes, if at all. We set up unified notation to describe all such QCDCL-based proof systems.

	\begin{definition}[$\Dord{\D}$]
		\label{def:depord}
		For a  dependency scheme $\D$, the decision policy $\Dord{\D}$ permits a decision on a variable $x$ at some point in a trail if all variables $y$ on which $x$ depends according to $\D$ (i.e.\ $(y,x)\in\D$), have already been assigned. 
	\end{definition}

	We define a new notation to describe QCDCL based proof systems introduced below with or without the option for cube learning. 
	
	\begin{definition}
		\label{def:QCDCLcubeDep}
                $\QCDCLDep{\Ord}{\ClausePolicy}{\CubePolicy}$ is the \QCDCL\ proof system where
		\begin{enumerate}
		\item $\Ord$ denotes the decision policy; e.g.\ $\lo$, $\Dord{\D'}$, $\ao$. \\
                          In $\Dord{\D'}$, $\D'$ denotes the dependency scheme, such as $\Dord{\Drrs}$, $\Dord{\Dstd}$. 
			\item $\ClausePolicy$ is the dependency scheme $\D$ used in reduction, propagation, and learning  for clauses. Note that this scheme need not be the same as the $\D'$ in $\Dord{\D'}$. 
			\item $\CubePolicy \in \{\NoCube, \CubeLD, \CubeD{\D}\}$ denotes the type of usage of cube learning;
			\begin{enumerate}
				\item $\NoCube$: No cube learning.
				\item $\CubeLD$: Cube Learning used, but no dependency scheme (only $\Dtrv$) in cube propagation, and cube learning using $\LDQTermRes$.
				\item $\CubeD{\D}$: Cube Learning used, dependency scheme  $\D$ used in propagation, and cube  learning is done using $\QDTermRes$. 
			\end{enumerate} 
		\end{enumerate}
	\end{definition}

	Note that in this notation,  clause learning always uses $\LDQDRes$, where $\D$ might well be $\Dtrv$.
        However, for cube learning, the  propagation and learning can use either long-distance term resolution \LDQDTermRes, or dependency schemes without long-distance   $\QDTermRes$, not both. We impose this condition  because the soundness of $\LDQDTermRes$ is not known. 

In the notation of \cref{def:QCDCLcubeDep}, the standard vanilla \QCDCL\ system denoted \QCDCLlored\    in \cite{BeyersdorffBohm-LMCS23} would be $\QCDCLDep{\lo}{\Dtrv}{\NoCube}$, whereas the \QCDCLcube\ system from \cite{BohmPB-AI24} would be $\QCDCLDep{\lo}{\Dtrv}{\CubeLD}$. Further,  the dependency-based system \QCDCLDrrs\ introduced in \cite{ChoudhuryMahajan-JAR24} would be  $\QCDCLDep{\lo}{\Drrs}{\NoCube}$.

	To define what a derivation in these \QCDCL\ systems must look like, we must first define trails and the learnable clauses and cubes from trails in this system. The following definitions are the natural generalisations of the corresponding ones from \cite{BeyersdorffBohm-LMCS23,BohmPB-AI24}. We give a short illustration in \cref{ex:illustration} after \cref{def:derivQCDCLcubeDep}. 
	
	The trail is a sequence of literals (or $\square$, $\top$)
	\begin{equation*}
		T = (p_{(0,1)}, \cdots , p_{(0,g_0)};\mathbf{d_1}, p_{(1,1)}, \cdots p_{(1,g_1)}; \mathbf{d_2}, \cdots  \cdots \cdots; \mathbf{d_r}, p_{(r,1)}, \cdots p_{(r,g_r)} )
	\end{equation*} 
	where the literals $d_i$ (in boldface) are decision literals, the literals $p_{i,j}$ are propagated literals. and no opposing literals appear. We can also view it as a set of literals or an assignment, and the corresponding clause (cube) is the disjunction (conjunction) of all literals in it.

	The learnable constraints from a   trail are defined as follows:
	
	\begin{definition}[learning from conflict]
		From a trail
		\begin{equation*}
			T = (p_{(0,1)}, \cdots , p_{(0,g_0)};\mathbf{d_1}, p_{(1,1)}, \cdots p_{(1,g_1)}; \mathbf{d_2}, \cdots  \cdots \cdots; \mathbf{d_r}, p_{(r,1)}, \cdots p_{(r,g_r)} )
		\end{equation*} 
		ending in a conflict $p_{(r,g_r)} \in \{\square,\top \}$,  the set $L_T$
		of learnable constraints has a clause $C_{(i,j)}$ associated with each propagated literal $p_{(i,j)}$ propagated
		in the trail if $p_{(r,g_r)}=\square$, and a cube associated with each propagated literal in the trail if $p_{(r,g_r)}=\top$. These associated clauses/cubes are constructed by tracing the
		conflict backwards through the trail as follows. ($\ante(\ell)$
		denotes the clause/cube that causes literal $\ell$ to be propagated; i.e.\ the antecedent.) Starting with $\ante(\square)$ (respectively $\ante(\top)$), we resolve in reverse over the antecedent clauses (cubes) that propagated the existential (universal) variables as described below. All such resulting clauses (cubes) are learnable constraints. 

In particular, if $p_{(r,g_r)}=\square$, then 
\begin{itemize}
\item $C_{(r,g_r)} = \reduceD(\ante(p_{(r,g_r)}))$.
\item For $i \in \{0,1, \cdots , r\}$ and $j \in [g_i -1]$,
  if $\var(p_{(i,j)})\in X_\exists$ and  $\bar{p}_{(i,j)} \in C_{(i,j+1)}$, then
  $C_{(i,j)}$ is the clause obtained by resolving the clause $C_{(i,j+1)}$ with the clause obtained from the antecedent of $p_{(i,j)}$ after reduction; such a resolution is possible on pivot $p_{(i,j)}$. Otherwise, $C_{(i,j)}$ is simply the same as $C_{(i,j+1)}$. Thus, $C_{(i,j)}$  equals \\
$\reduceD[\resolve(C_{(i, j+1)}, \reduceD(\ante(p_{(i,j)})), p_{(i,j)})]$
if $\var(p_{(i,j)})\in X_\exists$ and $\bar{p}_{(i,j)} \in C_{(i,j+1)}$, and is 
$C_{(i,j+1)}$ otherwise.
  
			
\item The learning process skips decision variables, so $C_{(i,g_i)}$ is defined using $p_{(i,g_i)}$ and $C_{(i+1,1)}$. For $i \in \{0,1, \cdots , r-1\}$.
$C_{(i,g_i)}$ equals $\reduceD[\resolve(C_{(i+1, 1)}, \reduceD(\ante(p_{(i,g_i)})), p_{(i,g_i)})]$ if  $\var(p_{(i,g_i)})\in X_\exists$ and $\bar{p}_{(i,g_i)} \in C_{(i+1,1)}$, and is $C_{(i+1,1)}$ otherwise.

\end{itemize}
		If $p_{(r,g_r)}=\top$, then non-trivial cube-resolution is performed when $p_{(i,j)}$ is universal, not existential. The set $L_T$ depends on $\CubePolicy$.  If $\CubePolicy=\CubeLD$, then  $\reduceD$ is replaced by $\reduce_\exists$.	If $\CubePolicy = \CubeD{\D}$, then  $\reduceD$ is replaced by $\reduceD_\exists$, but the $\resolve$ step is performed  only if it is a valid $\QDTermRes$ resolution step; otherwise we use the previously learnt cube, just as we do in clause learning when the resolution on $p_{(i,j)}$ is not defined.
	\end{definition}

	\begin{definition}[learning from satisfaction]
		From a trail $T$ 
		that assigns all variables, satisfies all clauses, and does not satisfy any cube, 
		the set of learnable constraints is defined as follows:
For any set $L$ of literals, let $t_L$ denote the cube that is the conjunction of all literals in $L$. Viewing the trail $T$ as a set of literals, 
                \[
                L_T = \left\{
                \begin{array}{ll}
                  \{\reduce_\exists(t_{T'}) \mid T' \subset T; \textrm{~$T'$ satisfies all axioms and learnt clauses}\}
                  &  \textrm{if~} \CubePolicy=\CubeLD \\
                  \{\reduceD_\exists(t_{T'}) \mid T' \subset T;  \textrm{~$T'$ satisfies all axioms and learnt clauses}\}
                  & \textrm{if~} \CubePolicy=\CubeD{\D}. 
                \end{array} \right.
                \]
		
	\end{definition}
	
	\begin{definition}
		\label{def:derivQCDCLcubeDep}
		 For a specific choice of $\Ord$, $\ClausePolicy$, $\CubePolicy$, let $\PS$ be the proof system $\QCDCLDep{\Ord}{\ClausePolicy}{\CubePolicy}$. A $\PS$-derivation $\iota$ from a PCNF QBF $\Phi = \calQ \vec{x} \cdot \varphi$ of a clause or cube  $C$ is a sequence $\iota$ of triples, 
		 $ \iota= (T_1, C_1, \pi_1), \cdots , (T_m,C_m, \pi_m)$, where each $T_i$ is a trail, each $C_i$ is a clause/cube, and $C_m=C$. The objects $T_i,C_i,\pi_i$ are as defined below.

                 For each $i\in[m]$, $\varphi_i$ is a propositional formula of the form 
$\varphi_i = \left( \bigwedge_{C\in\mathcal{C}_i} C\right) \vee \left(\bigvee_{T\in\mathcal{T}_i} T\right),$ where $\mathcal{C}_i$ is a set of clauses and $\mathcal{T}_i$ is a set 
		 of cubes. These formulas are  defined iteratively; initially we have all the clauses of $\varphi$ and no terms, and after each trail either a clause or a term is learnt and added. Formally,
		 \[
		\begin{array}{rrlrl}
			& \mathcal{C}_1 & = \{C \mid C\in  \varphi\}, &  \mathcal{T}_1 & = \emptyset\\
			\textrm{If $C_i$ is a clause:}
			& \mathcal{C}_{i+1} &= \mathcal{C}_i \cup\{C_i\}, 
			& \mathcal{T}_{i+1} &= \mathcal{T}_i.\\
			\textrm{If $C_i$ is a cube:}
			& \mathcal{C}_{i+1} &= \mathcal{C}_i, 
			& \mathcal{T}_{i+1} &= \mathcal{T}_i \cup \{C_i\}.
		\end{array}
		\]
For each $i\in[m]$, $\Phi_i$ is the QBF with the same quantifier prefix as $\Phi$, and inner formula $\varphi_i$. 
For each $i\in[m]$, $T_i$ is a trail from the formula $\Phi_i$,
$C_i$ is a learnable clause or cube from $T_i$, 
and 
$\pi_i$ is the 
derivation of $C_i$ from $\Phi$ in the system as per $\ClausePolicy$ or $\CubePolicy$.
		
		A refutation in these systems is a derivation of the empty clause $\square$, and a verification in this system is a derivation of the empty term $\top$.
	\end{definition}
	
\begin{example}
  \label{ex:illustration}
  The $\TwoPHPandCT_n$ formulas, defined in \cite{ChoudhuryMahajan-JAR24}(Section 4.5),  have the prefix\\
$\calQ = \forall u   \exists x_1\cdots  x_{s_n} ~~ \exists y_1 \cdots  y_{s_n} ~~\forall v  \exists z_1,z_2$ and the matrix 
  \[
  \begin{array}{lc}
&  u \vee \PHP_n(x_1, \cdots, x_{s_n}) \qquad \bar{u} \vee \PHP(y_1, \cdots, y_{s_n}) \\
&  v \vee z_1 \vee z_2 \ ,\  v \vee \bar{z}_1 \vee z_2 \ , \  v \vee z_1 \vee \bar{z}_2 \ , \ v \vee \bar{z}_1 \vee \bar{z}_2 \\ 
\end{array}
\]
Here $\PHP_n$ refers to the propositional Pigeon-hole-principle formulas that assert the existence of a map from $n+1$ pigeons to $n$ holes without collision; these formulas are known to be exponentially hard for resolution. Due to $\PHP_n$, the matrix is unsatisfiable, and thus cube-learning makes no difference. So we consider \NoCube.

Consider refutations using $\Dord{\D}$, with $\D=\Drrs$ or $\Dstd$. For these formulas, it can be seen that $\Drrs=\emptyset$
and $\Dstd = \{(u,x_i),(u,y_i):  \textrm{~for all~} i \} \cup \{(v,z_1), (v,z_2)\}$.  The  $v,z_1,z_2$ variables are completely  independent from the $u,x,y$ variables; neither $(x_i,v)$ nor $(y_i,v)$  is in $\D$ for any $i$. Therefore using the $\Dord{\D}$ decision policy, we can decide $v$ or $z_i$ in the beginning.  Consider the trail that decides to assign $v$ to false and then $z_1$ to true. The clause $v \vee \bar{z}_1 \vee  z_2$ becomes unit, so $z_2$ is propagated. Then the clause $v \vee \bar{z}_1 \vee \bar{z}_2$ becomes empty, and the empty clause is propagated, leading to conflict. That is, in $ T_1 =  \mathbf{\bar{v}; z_1} , z_2 , \square$,  we have 
$\ante(\square) = v \vee \bar{z}_1 \vee \bar{z}_2$, and 
$\ante(z_2) =  v \vee \bar{z}_1 \vee  z_2$. In conflict analysis, we first resolve (the reduced version of) $\ante(\square)$ with $\ante(z_2)$ to obtain $v \vee \bar{z}_1$.

If $\Drrs$ is used in learning, then this can be reduced further to $\bar{z}_1$, which is learnt, i.e.\ included in $\mathcal{C}_2$ and $\Phi_2$. 
The next trail then begins with the propagaed literal $\bar{z}_1$, and propagates further; 	$T_2 =  \bar{z}_1 , z_2 , \square $, and 
$\ante(\square) = v \vee {z}_1 \vee \bar{z}_2$, 
$\ante(z_2) =  v \vee {z}_1 \vee  z_2$,
$\ante(\bar{z}_1) =  \bar{z}_1$, 
allowing us to learn  $\square$. This is a refutation in
$\QCDCLDep{\Drrs}{\Drrs}{\NoCube}$ or
$\QCDCLDep{\Dstd}{\Drrs}{\NoCube}$, but not in
$\QCDCLDep{\lo}{\Drrs}{\NoCube}$.

If $\Dstd$ or $\Dtrv$ is used in learning, then from trail $T_1$ the clause $v\vee \bar{z}_1$ is learnt (i.e.\ included in $\mathcal{C}_2$ and $\Phi_2$) since it cannot be further reduced. The next trail must again begin with a decision. With $T_2 =  \mathbf{\bar{v}}, \bar{z}_1 , z_2 , \square $,
$\ante(\square) = v \vee {z}_1 \vee \bar{z}_2$,
$\ante(z_2) =  v \vee {z}_1 \vee  z_2$, 
$\ante(\bar{z}_1) = v \vee \bar{z}_1$, 
allowing us to learn $\square$. This is a refutation in
$\QCDCLDep{\D_1}{\D_2}{\NoCube}$, where $\D_1$ could be $\Drrs$ or $\Dstd$ but not $\Dtrv$, and $\D_2$ could be $\Dstd$ or $\Dtrv$. \qed
\end{example}

	By definition, $\Dord{\D}$  generalises  $\lo$, and $\ao$ generalises $\Dord{\D}$. Thus, 
	\begin{observation}
	  \label{obs:lodoao}
For a $\CubePolicy \in \{\NoCube, \CubeLD, \CubeD{\D}\}$,  and dependency schemes $\D$, $\D'$, 
\begin{itemize}
  \item Every derivation in $\QCDCLDep{\lo}{\D}{\CubePolicy}$ is a
    derivation in $\QCDCLDep{\Dord{\D'}}{\D}{\CubePolicy}$.
  \item Every derivation in $\QCDCLDep{\Dord{\D'}}{\D}{\CubePolicy}$ is a
    derivation in $\QCDCLDep{\ao}{\D}{\CubePolicy}$.
\end{itemize}
	\end{observation}
	
        Trails in a system without cube learning are also trails in the corresponding system with cube learning. 
        Hence:  	\begin{observation}
	  \label{obs:nocubeImplycube}
          For $\CubePolicy \in \{\CubeLD, \CubeD{\D}\}$, and for any decision policy  $\Ord$, 
		any derivation in $\QCDCLDep{\Ord}{\D}{\NoCube}$ is also a derivation in $\QCDCLDep{\Ord}{\D}{\CubePolicy}$.
	\end{observation}

If the matrix of the given PCNF formula is \textit{unsatisfiable}, then no satisfying trail can ever be constructed, no matter what policy is used, so no "cube learning" can happen. Hence:
	\begin{observation}
		\label{obs:unsatnocubeEqcube}
		For $\CubePolicy \in \{\CubeLD, \CubeD{\D}\}$,   and for any decision policy  $\Ord$,
                if a PCNF formula $\Phi$ has an unsatisfiable matrix, then any derivation from $\Phi$ in the system $\QCDCLDep{\Ord}{\D}{\CubePolicy}$  is also a derivation in $\QCDCLDep{\Ord}{\D}{\NoCube}$.        
	\end{observation}

        The following result is shown in \cite{ChoudhuryMahajan-JAR24}.
        \begin{proposition}[Theorem 1  in \cite{ChoudhuryMahajan-JAR24}]\label{prop:lo-d-no-cube-complete}
          For any normal dependency scheme $\D$, the system $\QCDCLDep{\lo}{\D}{\NoCube}$ is refutationally complete.
        \end{proposition}
        

	
		

A minor adaptation of the proof of Theorem 3.9 in \cite{BeyersdorffBohm-LMCS23} shows  the following soundness: 
	\begin{theorem}
		\label{thm:qcdclao-sound}
		For any normal dependency scheme $\D$, the proof systems $\QCDCLDep{\ao}{\D}{\CubeLD}$ and $\QCDCLDep{\ao}{\D}{\CubeD{\D}}$ are sound. 
	\end{theorem}
	
	\begin{proof}
		
		
          To show that both these systems are sound, it is enough to show the following three statements: (1)~The derivation of any learnt clause is a valid $\LDQDRes$ derivation.  
          (2)~If $\CubePolicy = \CubeLD$, the derivation of any learnt cube is a valid $\LDQTermRes$ derivation, and the addition of cubes when learning from satisfaction is sound. (3)~If $\CubePolicy = \CubeD{\D}$, the derivation of any learnt cube is a valid $\QDTermRes$ derivation. From these three it follows that sticking together the derivations of the final learnt empty clause/term  gives a proof in the corresponding system, and all these systems are known to be sound. 

          Statement~(3) is true by definition: term resolution in learning is performed only if it is valid in $\QDTermRes$. For Statement~(2),  cube learning is shown to be sound in \cite[Theorem 3.8]{BohmPB-AI24}.
          For statement~(1), we need to show that the resolution steps performed while learning respect the side-conditions of $\LDRD{\D}$. The analogous statement when $\D=\Dtrv$ is proved in  \cite[Lemma 3.7, Proposition 3.8, Theorem 3.9]{BeyersdorffBohm-LMCS23}, but the same proof works with any  $\D$. It is formally shown in \cref{lem:soundness} below.
	\end{proof}

\begin{restatable}{lemma}{restateLemmaSoundness}
		\label{lem:soundness}
		For any normal dependency scheme $\D$, the derivations of a clause learnt from a trail in the proof systems $\QCDCLDep{\ao}{\D}{\CubeLD}$ and $\QCDCLDep{\ao}{\D}{\CubeD{\D}}$ are valid $\LDQDRes$ derivations.
	\end{restatable}
	
\begin{proof}
This proof essentially replicates the proofs of Lemma 3.7 and Proposition 3.8 from \cite{BeyersdorffBohm-LMCS23}.
		We need to show that every clause learnt from a trail :
		$$\T = (p_{(0,1)}, \cdots , p_{(0,g_0)};\mathbf{d_1}, p_{(1,1)}, \cdots p_{(1,g_1)}; \mathbf{d_2}, \cdots  \cdots \cdots; \mathbf{d_r}, p_{(r,1)}, \cdots p_{(r,g_r)} )$$
		is a valid \LDQDRes\ derivation. Let $C_{i,j}$ denote the clause learnt corresponding to propagated literal $p_{i,j}$. 
		
		\noindent	\textbf{Step 1:}  No $C_{i,j}$ contains an existential tautology. \\
		Suppose there exists a variable $x$ such that $x \neq \var(p_{i,j})$ and $x \in C_{i,j+1}$ and $ \bar{x} \in  \reduceD(\ante(p_{i,j}))$. Let $A = \ante(p_{i,j})$, then since $x$ is existential variable, and  $\bar{x} \in A$, therefore $x$ must be assigned in the trail prior to the propagation of $p_{i,j}$.\\
		On the other hand, we have $x \in C_{i,j+1}$, which  is the learnable clause which is derived with the aid of antecedent clauses of literals occurring right of $p_{i,j}$ in the trail. In particular, we can find some $p_{k,m}$ right of $p_{i,j}$ in the trail with  $x \in \ante(p_{k,m})$. But because $x$ appears in the trail to the left of $p_{i,j}$, this gives a contradiction since $ \ante(p_{k,m})$ must not become true before propagating $p_{k,m}$.
		
		\noindent	\textbf{Step 2:} Derivation of clauses with universal tautolgies is sound.\\
		Proof goes via contradiction. Suppose there exists a universal tautology derived which is unsound. Without loss of generality let that variable be $u$ and the propagated literal 
		over which this resolution happens be $p_{i,}$. Since the resolution is unsound $(u,\var(p_{i,j})) \in \D$ and on of the following conditions must hold:
		\begin{enumerate}
			\item $u \in C_{i,j+1}$ and $\bar{u} \in \ante(p_{i,j})$
			\item  $u \vee \bar{u} \in C_{i,j+1}$ and $\bar{u} \in \ante(p_{i,j})$
			\item  $u \in C_{i,j+1}$ and $u \vee \bar{u} \in \ante(p_{i,j})$
			\item $u  \vee \bar{u} \in C_{i,j+1}$ and $u \vee \bar{u} \in \ante(p_{i,j})$
		\end{enumerate}
		
		Consider the first case: Since $u \in C_{i,j+1}$ there has to be a propagated literal $p_{k,m}$ right of $p_{i,j}$ in the trail such that $u \in \ante(p_{k,m})$. In order to become unit, the $u$ in $\ante(p_{k,m})$ needs to vanish. We distinguish two cases:\\
		Case (i):  $\bar{u}$ was assigned before $p_{k,m}$ was propagated. Then $\bar{u}$ does not appear in the trail, then for $p_{i,j}$ to be propagated $\bar{u}$ must have been reduced in $\ante(p_{i,j})$ which is possible only if $(u,\var(p_{i,j})) \not\in \D$ giving rise to a contradiction.\\
		Case (ii): $u \in \ante(p_{k,m})$ is removed via reduction. For propagations $p_{i,j}$, $p_{k,m}$ to both happen $u$, $\bar{u}$ could not be assigned in the trail prior to propagating $p_{i,j}$, therefore  for $p_{i,j}$ to be propagated $\bar{u}$ must have been reduced in $\ante(p_{i,j})$ which is possible only if $(u,\var(p_{i,j})) \not\in \D$ giving rise to a contradiction. 
		
		The same argument above works for all the remaining cases.
	\end{proof}


        Thus using dependency schemes in decision order gives sound and complete systems.
	\begin{theorem}\label{thm:depord-sound-complete}
           For any  dependency schemes $\D'$, $\D$ where $\D$ is normal, 
          and for each $\CubePolicy\in\{\NoCube,\CubeLD,\CubeD{\D}\}$,  
          the proof system $\QCDCLDep{\Dord{\D'}}{\D}{\CubePolicy}$ is sound and refutationally complete.
	\end{theorem}
        \begin{proof}
        	 By \cref{prop:lo-d-no-cube-complete},  $\QCDCLDep{\lo}{\D}{\NoCube}$ is refutationally complete. By \cref{thm:qcdclao-sound}, $\QCDCLDep{\ao}{\D}{\CubePolicy}$ is sound.
        	 By \cref{obs:lodoao} and \cref{obs:nocubeImplycube}, all the aforementioned systems are sound and refutationally complete.
%
%
%
        \end{proof}

	\subsection{Strength of \QCDCL\ based proof systems with $\Dord{\D}$}
	\label{subsec:depordAdv}
	
	In \QCDCL\  based proof systems, incorporating dependency schemes into propagation and learning processes does not always yield benefits: as shown in \cite{ChoudhuryMahajan-JAR24},  certain pathological formulas can render the addition of dependency schemes disadvantageous when decisions are constrained to the $\lo$ decision policy. 

        If the dependency shceme is allowed to influence the decision order (i.e., the system adopts the $\Dord{\D}$  decision policy), we show below that the resulting systems are strictly more powerful than their counterparts using $\lo$. (For $\D_1=\Dstd$, an advantage over $\lo$ was noted already in \cite{Lonsing-Thesis12}. ) However, they remain strictly weaker than the \LDQDRes\ systems.

	\begin{restatable}{theorem}{restatedepordlevord}
	  \label{thm:depord-levord}
          For dependency schemes $\D_1 \in \{\Drrs,\Dstd\}$  and 	$\D_2 \in \{\Drrs,\Dstd,\Dtrv\}$, and
          for policy $\CubePolicy\in  \{\NoCube, \CubeLD, \CubeD{\D_2} \}$, 
	the proof system $\QCDCLDep{\Dord{\D_1}}{\D_2}{\CubePolicy}$ $p$-simulates  $\QCDCLDep{\lo}{\D_2}{\CubePolicy}$ and is not simulated by it. 
	\end{restatable}	
\begin{proof}
  The $p$-simulation follows from \cref{obs:lodoao}.
  
  For $\D_1=\Dstd$, an advantage over $\lo$ was noted already in \cite{Lonsing-Thesis12}.

  For $\D_1=\Drrs$, to show that there is no reverse simulation, we consider the $\TwoPHPandCT$ formulae defined in \cite{ChoudhuryMahajan-JAR24}, and described in \cref{ex:illustration}. 
. These have an unsatisfiable matrix, so by \cref{obs:unsatnocubeEqcube}, it suffices to show lower and upper bounds for $\CubePolicy=\NoCube$. 
	
	For these formulas,
	$\Drrs=\emptyset$ and 
	$\Dstd = \{(u,x_i),(u,y_i):  \textrm{~for all~} i \} \cup \{(v,z_1), (v,z_2)\}$.
	
	It is shown in  \cite{ChoudhuryMahajan-JAR24} (Lemma 5) that these formulas  require exponential size refutations in $\QCDCLDep{\lo}{\Drrs}{\NoCube}$ and $\QCDCLDep{\lo}{\Dtrv}{\NoCube}$. Furthermore, the following  observation  shows that they also  require exponential size refutations in $\QCDCLDep{\lo}{\Dstd}{\NoCube}$; for these formulas, 
	$\Dstd = \{(u,x_i),(u,y_i):  \textrm{~for all~} i \} \cup \{(v,z_1), (v,z_2)\}$. 
	With the $\lo$ decision policy, the first decision must be on $u$, which causes no propagations, and subsequent decisions in $\lo$ force  refuting $\PHP$ (in either $x$ or $y$), which is known to requires exponential size.
	
	On the other hand, we have already seen in \cref{ex:illustration} that they have short (constant-sized) refutations in $\QCDCLDep{\Dord{\D_1}}{\D_2}{\NoCube}$ if $\D_1=\Drrs$ or $\Dstd$.
\end{proof}

	\begin{restatable}{theorem}{restatedepordldqdres}	
	\label{thm:ldqd-qcdcldord}
		For $\CubePolicy \in \{\NoCube,\CubeLD, \CubeD{\D}\}$ in  $\D \in \{ \Dtrv, \Dstd, \Drrs\}$, the proof system $\LDQDRes$ $p$-simulates $\QCDCLDep{\Dord{\D}}{\D}{\CubePolicy}$ and is not simulated by it.
        \end{restatable}
We defer the proof of this theorem to the next section, since it uses a new formula that we define there, the \DoubleLongEq\ formulas. 
		
	\section{Dependency-schemes-based QCDCL systems restricted to $\lo$ }
	\label{sec:QCDCLcuberesults}
	
        We now restrict our attention to proof systems utilizing only the $\lo$ decision policy.  Even with $\lo$, dependency schemes can affect (enhance or impair) the performance of \QCDCL\ systems. These are the systems also considered in \cite{BeyersdorffBohm-LMCS23,BohmPB-AI24,ChoudhuryMahajan-JAR24}.

	\subsection{Preprocessing as a tool}
	\label{subsec:preprocessing}
        In \cite{ChoudhuryMahajan-JAR24}, another way of incorporating dependency schemes was also considered, through ``preprocessing''. For a more complete comparison, we very briefly define such systems here as well.
Preprocessing a formula using a dependency scheme simply means applying all reductions enabled by it on the input formula, and then proceeding with whatever version of \QCDCL\ is of interest, on the reduced formula.

	\begin{definition}
	\label{def:d+qcdcldep system}
	For a QBF $\Phi = \calQ \cdot \phi$ and a normal dependency scheme $\D$,
        a derivation of a clause $C$ from $\Phi$ in 
        $\D + \QCDCLDep{\Ord}{\ClausePolicy}{\CubePolicy}$ 
		is a derivation of $C$ from the QBF $\Psi = \reduceD(\Phi)$ in  $\QCDCLDep{\Ord}{\ClausePolicy}{\CubePolicy}$.
	\end{definition}
	
        As shown in Theorem~4 of \cite{ChoudhuryMahajan-JAR24}, preprocessing using  $\Drrs$ can significantly alter the system strength. In contrast, we observe below that  preprocessing using schemes $\Dtrv$ or $\Dstd$ has no effect, no matter what version of \QCDCL\ is the subsequent system.

\begin{restatable}{proposition}{restateDtrvDstdnopre}

        \label{prop:DtrvDstd-nopre}
        For every decision policy $\Ord$, dependency scheme $\D  \in \{\Dtrv, \Dstd, \Drrs\}$ and for $ \CubePolicy \in \{\NoCube,\CubeLD, \CubeD{\D}\}$, the proof systems
%
 $\QCDCLDep{\Ord}{\D}{\CubePolicy}$, 
 $\Dtrv + \QCDCLDep{\Ord}{\D}{\CubePolicy}$, and
 $\Dstd + \QCDCLDep{\Ord}{\D}{\CubePolicy}$ 
        are equivalent to each other.
\end{restatable}
\begin{proof}
	Let $\Phi$ be any given PCNF formula,  and let $\Psi = \reduce(\Phi)$.
	By definition, $\Dstd(\Phi)\subseteq\Dtrv(\Phi)$.
	So all reductions permitted by $\Dtrv$ are also permitted by $\Dstd$. If a clause $C$ of $\Phi$ has variables $x,y$ with $(x,y)\in \Dtrv$ , then by definition of $\Dstd$, $(x,y)$ is also in $\Dstd(\Phi)$. So $\Dstd$ does not enable any new reductions on initial clauses. 
	Hence  $\Psi=\reduceDstd(\Phi) = \reduceD^{\mathtt{trv}}(\Phi)$, so the second and third proof systems are equivalent.
	
	For any clause $C$, $\reduce(C)$ can only remove universal
	variables from $C$. By the way $\Psi$ is defined,  if $(x,y)\in \D(\Psi)$, then it is also in $\D(\Phi)$. In the other direction, if $(x,y)\in\D(\Phi)$, and if both $x,y$ appear in the matrix of $\Psi$, then $(x,y)$ is also in $\D(\Psi)$ because the prefix of $\Phi$ and $\Psi$ is the same, and the  witnessing sequence (in the case of $\Dstd$ or $\Drrs$) has only existential literals which are not removed, so the same sequence is a witness  in $\Psi$ too. 
	
	To show that the first and second proof systems are the same, note that for $\D \in  \{\Dtrv, \Dstd, \Drrs\}$, for any clause $C$, $\reduceD(C) = \reduceD(\reduce(C))$. Therefore if one of  $\Dtrv$, $\Dstd$ or $\Drrs$ are used in propagation and learning, then all the propagations in the first trail of either refutation are also enabled in  the other. (To be pedantic, a derivation in $\Phi$ may have universal variables that have vanished from the matrix of $\Psi$ but are still in the quantifier prefix of $\Psi$; these can have no effect on any propagation since they vanished through applications of $\reduce$.) Hence  the first clause/cube learnt in any one system can also be learnt in the other. Continuing this argument on subsequent trails, the entire derivation can be replicated. 
\end{proof}

	\subsection{Some New Formulae}
	\label{subsec:formulae}
	
We now introduce some new formulas which will be used to pinpoint the relative strengths of the proof systems.

\paragraph*{The $\DoubleLongEq_n$ formulas.}
The \Equality\ formulas, first defined in \cite{BeyersdorffBH-LMCS19}, 
show that the proof system \QCDCL\ with cube learning is stronger than \QCDCL\ without 
	 \cite{BohmPB-AI24}. We wish to show that cube-learning offers a similar advantage for systems with $\Drrs$. 
	 The \Equality\ formulas cannot show this because $\Drrs(\Equality)=\emptyset$; so using $\Drrs$ in any way makes them easy to refute irrespective of cube-learning. To achieve the desired separation, we modify the \Equality\ formula by adding two clauses that make $\Drrs$ and $\Dtrv$ identical. These new formulas, called \DoubleLongEq, maintain the separation without altering the hardness of \Equality.
		
		\begin{restatable}{formula}{restateDoubleLongEq}
			\label{for:DoubleLongEq}
			The $\DoubleLongEq_n$ formula has the prefix \\
			$\exists x_1  \cdots x_n ~ \forall u_1 \cdots u_n ~ \exists t_1 \cdots t_n$ and the PCNF matrix 
			\[ 
			\begin{array}{lc}
				&
				 \underbrace{(\bar{t}_1 \vee \cdots \vee \bar{t}_n)}_{T_n} ~ \wedge ~  \bigwedge_{i=1}^n
				\left[
				\underbrace{(x_i \vee u_i \vee t_i)}_{A_i}
				\wedge
				\underbrace{(\bar{x}_i \vee \bar{u}_i \vee t_i)}_{B_i}
				\right] ~ \wedge ~  
				\\
				& \underbrace{ (\bar{u}_1 \vee \cdots \bar{u}_n \vee \bar{t}_1 \vee \cdots \bar{t}_n) }_{UT_n}  ~ \wedge ~  \underbrace{ (\bar{u}_1 \vee \cdots \bar{u}_n \vee t_1 \vee \cdots t_n) }_{UT'_n}  
			\end{array}
			\]                                
		\end{restatable}
	
		(Note: deleting the clauses $UT_n$ and $UT'_n$ gives the \Equality\ formulas.)

\begin{restatable}{proposition}{restateDepDoubleLongEq}
\label{clm:DoubleLongEqDrrs}                  
  For $\Phi=\DoubleLongEq$, $\reduce(\Phi)= \Phi$, and
  $\Drrs(\Phi) = \Dstd(\Phi) = \Dtrv(\Phi)$.
\end{restatable}
\begin{proof}     
	By definition,  $\Dtrv(\DoubleLongEq) = \{(u_i,t_j) : i,j \in \{1 \cdots n \}\}$; that is, each $t_j$ variable depends on each $u_i$ variable. Since each occurrence of a $u$ variable in the formula is blocked by some $t$ variable, we have $\reduce(\DoubleLongEq)= \DoubleLongEq$.

	The next claim is that for this family of formulae $\Phi$, $\Drrs(\Phi) = \Dstd(\Phi) = \Dtrv(\Phi)$.
	It suffices to show that $\Dtrv(\Phi)\subseteq \Drrs(\Phi)$.
	
	We want to show that each $t_j$ depending on each $u_i$  is the case for $\Drrs$ as well. We consider two cases. 
	
	Case 1: $i=j$. The clauses $A_i$ and $UT_n$ contain the resolution paths $(u_i,t_i)$ and $(\bar{u}_i,\bar{t}_i)$ respectively. Therefore $(u_i,t_i) \in \Drrs$ for all $i \in \{1 \cdots n \}$. 
	
	Case 2: $i \neq j$.  The sequence of clauses $A_i, UT_n$ contains the resolution path $(u_i,t_i),(\bar{t}_i, \bar{t}_j)$, while the clause $UT'_n$ conatins the resolution path $(\bar{u}_i,t_j) $. Therefore $(u_i,t_j) \in \Drrs$ for all $i \neq j \in \{1 \cdots n  \}$. 	
\end{proof}

This makes the usage of the dependency schemes $\Drrs$ or $\Dstd$ completely useless. We now show that the formulas are easy to refute with cube-learning, but hard if cube-learning is switched off. 
Both these results closely mirror the corresponding results for \Equality\ shown in \cite{BeyersdorffBohm-LMCS23} and \cite{BohmPB-AI24} respectively.

\begin{restatable}{lemma}{restateDoubleLongEqeasy}
	\label{lem:DoubleLongEq-easyQCDCLcube}
		For $\D_1,\D_2 \in \{\Dtrv,\Dstd, \Drrs\}$  and $\CubePolicy \in \{\CubeLD,\CubeD{\D_2}\}$,  the \DoubleLongEq\ formulas have polynomial size refutations in  $\D_1 + \QCDCLDep{\lo}{\D_2}{\CubePolicy}$

\end{restatable}	
To prove this, we first show that these formulas are easy to refute with cube-learning.
\begin{proposition}
	\label{prop:DoubleLongEq-easyQCDCLcube}
	The \DoubleLongEq\ formulas have polynomial size refutations in the proof system  $\QCDCLDep{\lo}{\Dtrv}{\CubeLD}$.
\end{proposition}
\begin{proof}
	The polynomial size refutation for these formulas in $\QCDCLDep{\lo}{\Dtrv}{\CubeLD}$ is exactly the same as the refutation in $\QCDCLDep{\lo}{\Dtrv}{\CubeLD}$for the  \Equality\ formulas, as described in \cite{BohmPB-AI24}. By constructing trails in exactly the same manner, we first learn $2n-2$ cubes of the form $(x_i \wedge \bar{u}_i)$ and $(\bar{x}_i \wedge u_i)$ for $i=1...n-1$ and then start clause learning by constructing trails ending in a conflict. The two new clauses $UT_n$ and $UT'_n$ play no role whatsoever. For completeness, we reproduce the entire refutation below; a reader familiar with the construction from \cite{BohmPB-AI24} can completely skip these details. 
	
	The proof goes in two stages. The first stage involves learning the cubes $x_i \wedge \bar{u}_i$ and $\bar{x}_i \wedge u_i$ for $i \in [n-1]$. The first trail is the following.
	$$ \mathcal{T}_1 = \mathbf{x_1;\cdots ;x_n;\bar{u}_1;\cdots ;\bar{u}_n;\bar{t}_1;t_2;\cdots;t_n}$$
	
	It assigns all variables without conflict and satisfies the matrix. The partial assignment $x_1 \wedge \bar{u}_1 \wedge \bar{t}_1 \wedge t_2 \wedge \cdots \wedge t_n$ contained in it is also a satisfying assignment, and reducing it with $\reduce_\exists$ we can learn the cube $x_1 \wedge \bar{u}_1$ from this trail.
	
	Analogously, creating a complementary trail $\mathcal{T}_1'$ where each decision is the complement of the decision in $\mathcal{T}_1$, we can learn the cube $\bar{x}_1 \wedge u_1$.
	
	Suppose we have learn $2i$ cubes in the same manner; $x_j\wedge \bar{u}_j$ and $\bar{x}_j\wedge u_j$ for $j=1,\cdots,i$. For $i+1$, create the following trail.
	$$\mathcal{T}_{i+1}= \mathbf{x_1},u_1,t_1; \cdots ; \mathbf{x_i},u_i,t_j; \mathbf{x_{i+1};\cdots;x_n;\bar{u}_{i+1};\cdots ; \bar{u}_n; \bar{t}_{i+1}; t_{i+2};\cdots;t_n} $$
	
	In this trail,  for $j \leq i$, 
	$ \ante(u_j) = x_j \wedge \bar{u}_j $ and  
	$ \ante(t_j) = \bar{x}_j \vee \bar{u}_j \vee t_j $.
	As earlier, the trail satisfies all clauses without conflict. Extracting the partial assignment  $x_{i+1} \wedge \bar{u}_{i+1} \wedge t_1 \wedge \cdots t_i \wedge \bar{t}_{i+1} \wedge t_{i+2} \wedge \cdots \wedge t_n$ which also satisfies the matrix, and reducing it, we can learn the cube $x_{i+1} \wedge u_{i+1}$.  Analogously through a trail $\mathcal{T}_{i+1}'$ we learn $\bar{x}_{i+1} \wedge u_{i+1}$.
	
	Having learnt the $2n-2$ cubes in this manner, we start with clause learning, where we proceed by constructing the trails 	
	$\mathcal{U}_{n-1}, \mathcal{V}_{n-1},  \mathcal{U}_{n-2}, \mathcal{V}_{n-2}, \cdots , \mathcal{U}_1, \mathcal{V}_1$ described below, and learn clauses $L_{n-1}, R_{n-1}, \cdots L_1, R_1$ corresponding to these trails.  We use $T_j$ to denote the subclause of $T_n$ with literals $\bar{t_i}$ for $i\;e j$.\\
	The initial trail is 
	\begin{equation*}
		\mathcal{U}_{n-1}= (\mathbf{x_1},u_1,t_1; \mathbf{x_2},u_2, t_2; \cdots ; \mathbf{x_{n-1}}, u_{n-1}, t_{n-1}, \bar{t}_n, x_n, \square)
	\end{equation*}
	The antecedent clauses are as follows: 
	\begin{eqnarray*}
		\ante(u_j) &=& x_j \wedge \bar{u}_j \\
		\ante(t_j) &=& \bar{x}_j \vee \bar{u}_j \vee t_j  \\
		\ante(\bar{t}_n) &=& T_n \\
		\ante(x_n) &=& x_n \vee u_n \vee t_n \\
		\ante(\square) &=&   \bar{x}_n \vee \bar{u}_n \vee t_n 
	\end{eqnarray*}
	From these  clauses we learn the clause $L_{n-1} = \bar{x}_{n-1} \vee  \bar{u}_{n-1} \vee (u_n \vee \bar{u}_n) \vee T_{n-2}$. \\
	Then we restart and create a symmetric trail to $\mathcal{U}_{n-1}$:
	\begin{equation*}
		\mathcal{V}_{n-1}= (\mathbf{\bar{x}_1}, \bar{u}_1,t_1; \mathbf{\bar{x}_2}, \bar{u}_2, t_2; \cdots ; \mathbf{\bar{x}_{n-1}}, \bar{u}_{n-1}, t_{n-1}, \bar{t}_n, x_n, \square)
	\end{equation*}
	where the antecedent clauses are 
	\begin{eqnarray*}
		\ante(\bar{u}_j) &=& \bar{x}_j \wedge u_j \\
		\ante(t_j) &=& x_j \vee u_j \vee t_j  \\
		\ante(\bar{t}_n) &=& T_n \\
		\ante(x_n) &=& x_n \vee u_n \vee t_n \\
		\ante(\square) &=&  \bar{x}_n \vee \bar{u}_n \vee t_n . 
	\end{eqnarray*}
	From this trail we can learn the clause $R_{n-1} = x_{n-1} \vee u_{n-1} \vee (u_n \vee \bar{u}_n)\vee T_{n-2}$.
	
	For $i$ in the range of $2$ to $n-1$, we define the following clauses:
	\begin{equation*}
		\begin{split}
			L_i = \bar{x}_i \vee \bar{u}_i \vee \bigvee_{j=i+1}^n (u_j \vee \bar{u}_j) \vee T_{i-1}\\
			R_i = x_i \vee u_i \vee \bigvee_{j=i+1}^n (u_j \vee \bar{u}_j)  \vee T_{i-1}\\
		\end{split}
	\end{equation*}        
	We claim that from the trail $\mathcal{U}_i$ we learn the clause $L_i$ and from the trail $\mathcal{V}_i$ we learn the clause $R_i$. 
	We have already established this  for $i = n-1$.  Suppose we have already learnt $L_{n-1}, R_{n-1}, ..... , L_{i+1}, R_{j+1}$ for $1 \leq i < n-1$. Continuing, we consider the next trail, 
	\begin{equation*}
		\mathcal{U}_{i}= (\mathbf{x_1}, u_1, t_1; \mathbf{x_2},u_2,t_2; \cdots ; \mathbf{x_{i}}, u_i, t_{i}, x_{i+1}, \square)
	\end{equation*} 
	where the antecedent clauses are as follows.
	\begin{eqnarray*}
		\ante(u_j) &=& x_j \wedge \bar{u}_j \\
		\ante(t_j) &=& \bar{x}_j \vee \bar{u}_j \vee  t_j  \\
		ante(x_{i+1}) &=& L_{i+1} \\
		ante(\square) &=& R_{i+1} 
	\end{eqnarray*}
	From this we learn the clause $L_i = \bar{x}_i \bar{u}_i \vee \bigvee_{j=i+1}^n (u_j \vee \bar{u}_j) \vee T_{i-1}$. \\
	Next we create the symmetrical trail, 
	\begin{equation*}
		\mathcal{V}_{i}= (\mathbf{\bar{x}_1}, \bar{u}_1, t_1; \mathbf{\bar{x}_2},\bar{u}_2,t_2; \cdots ; \mathbf{\bar{x}_{i}}, \bar{u}_i, t_{i}, x_{i+1}, \square)
	\end{equation*} 
	and the antecedent clauses are as follows:
	\begin{eqnarray*}
		\ante(\bar{u}_j) &=& \bar{x}_j \wedge u_j \\
		\ante(t_j) &=& x_j \vee u_j \vee t_j  \\
		\ante(x_{i+1}) &=& L_{i+1}  \\
		\ante(\square) &=& R_{i+1}  
	\end{eqnarray*}
	From this we can learn the clause $R_i = x_i \vee u_i \vee \bigvee_{j=i+1}^n (u_j \vee \bar{u}_j) \vee T_{i-1}$.
	
	The proof ends with the two trails 
	\begin{equation*}
		\mathcal{U}_1 = (\mathbf{x_1},u_1, t_1,x_2,\square)
	\end{equation*}
	with antecedents clauses 
	\begin{eqnarray*}
		\ante(u_1) = x_1 \wedge \bar{u}_1 \\
		\ante(t_1) = \bar{x}_1 \vee t_1  \\
		\ante(x_2) = L_{2} \\
		\ante(\square) = R_{2} 
	\end{eqnarray*}
	allowing us to learn the clause $L_1 = \bar{x}_1$, and finally the last trail 
	\begin{equation*}
		\mathcal{V}_1 = (\bar{x}_1, \bar{u}_1, t_1, x_2, \square)
	\end{equation*}
	with antecedent clauses
	\begin{eqnarray*}
		\ante(\bar{x}_1) = \bar{x}_1 \\
		\ante(\bar{u}_1) = \bar{x}_1 \wedge u_1 \\
		\ante(t_1) = \bar{x}_1 \vee u_1 \vee t_1 \\
		\ante(x_2) = L_2  \\
		\ante(\square) = R_2 
	\end{eqnarray*}
	
	Resolving over all propagations in this trail, we learn the empty clause, completing the refutation.                  
\end{proof}

Now proving \cref{lem:DoubleLongEq-easyQCDCLcube} is straightforward:
\begin{proof}(of \cref{lem:DoubleLongEq-easyQCDCLcube}.)
	It can be seen that the cubes learnt in the refutation described in \cref{prop:DoubleLongEq-easyQCDCLcube} require no cube learning via resolution steps; they are all learnt from trails ending in satisfaction, using the term axiom rule and the $\reduce_\exists$ rule.  Therefore for this particular refutation, every cube learning step in the $\QCDCLDep{\lo}{\Dtrv}{\CubeLD}$ refutation is also a valid step in a  $\QCDCLDep{\lo}{\Dtrv}{\CubeD{\Dtrv}}$ refutation.
	
	By the discussion after the formula definitions, these are also valid refutations in 
	$\D_1 + \QCDCLDep{\lo}{\D_2}{\CubeD{\D_2}}$  where $\D_1,\D_2 \in \{\Dtrv,\Dstd, \Drrs\}$.
\end{proof}

We now turn to hardness. 
\begin{restatable}{lemma}{restateDoubleLongEqhardQCDCL}
	\label{lem:DoubleLongEq-hardQCDCL}
	For  $\D_1,\D_2 \in \{\Dtrv,\Dstd, \Drrs\}$ the \DoubleLongEq\ formulas require exponential size refutations in $\D_1 + \QCDCLDep{\lo}{\D_2}{\NoCube}$.
\end{restatable}
\begin{proof}
                  By the discussion above, it suffices to show that the formulas require exponential size refutations in $\QCDCLDep{\lo}{\Dtrv}{\NoCube}$.
				
		  In \cite{BohmBeyersdorff-JAR23}, the authors consider $\Sigma^3$ formulas with a specific structure, called $XUT$-formulas with the $XT$-property. They introduce a semantic measure called gauge for $\Sigma^3$ QBFs, and show that for an $XUT$-formula with the $XT$-property, the size of a refutation in  $\QCDCLDep{\lo}{\Dtrv}{\NoCube}$ is at least exponential in its gauge.

                  The \DoubleLongEq\ formulas are easily seen to be $XUT$-formulas with the $XT$-property. We show now that they have linear gauge,  implying that they are exponentially hard.

Recall the definition of $XUT$ formulas and the gauge measure:
\begin{definition}[$XT$-property, \cite{BohmBeyersdorff-JAR23}]
	Let $\Phi$ be a PCNF QBF of the form $\exists X \forall U \exists T \cdot \phi$, where $X,U,T$ are non-empty sets of variables. Then $\Phi$ is an $XUT$-formula. We call a clause $C$  an
	\begin{itemize}
		\item $X$-clause: if it is non-empty and contains only $X$ variables,
		\item $T$-clause:  if it is non-empty and contains only $T$ variables,
		\item $XT$-clause: if it contains no $U$ variable and at least one $X$ and one $T$ variable, 
		\item $XUT$-clause: if it contains atleast one each of $X$,$U$, and $T$ variables. 
	\end{itemize} 
	$\Phi$ is said to fulfill the $XT$-property if $\phi$ contains no $XT$-clauses or unit $T$ clause,  and if no two $T$ clauses in $\phi$ are resolvable (the resolvent of any two $T$ clauses, if defined, is tautological). 
\end{definition}
\begin{definition}[gauge, \cite{BohmBeyersdorff-JAR23}]
	Let $\Phi$ be an $XUT$ formula. 
	The gauge of $\Phi$ is the size of the narrowest $X$-clause derivable  using only reductions and resolutions over variables in $T$. 
\end{definition}

First observe that none of the axioms of \DoubleLongEq\ are $X$-clauses. Therefore to derive an $X$-clause, there has to be some $T$-resolutions. A first $T$-resolution must involve either $T_n$ or $UT_n$, since only these clauses have $T$ variables negated. However, both these  clauses have all $n$ $T$-variables. Thus to eventually derive an $X$-clause,  there must be a resolution on every $t_i$ variable. Each such resolution introduces an $x_i$ variable. Therefore by the time all $T$ variables are removed, all the $X$ variables are  introduced. Therefore, the gauge of $\DoubleLongEq_n$ is  $n$.
\end{proof}

With this hardness result about \DoubleLongEq, we  can now complete the proof of
\cref{thm:ldqd-qcdcldord}. 
\begin{proof}(of \cref{thm:ldqd-qcdcldord}.)
  By definition, a valid $\LDQDRes$ refutation is contained within every $\QCDCLDep{\Dord{\D}}{\D}{\CubePolicy}$ refutation.  

We now show that the simulation is strict. 
We first prove it for the case when $\CubePolicy$ is $\NoCube$, using the
\DoubleLongEq\ formulas. Then we tweak the formulas slightly to extend the result to other cube-learning policies.

We saw in \cref{lem:DoubleLongEq-hardQCDCL}  that these formulas are hard for $\QCDCLDep{\Dord{\D}}{\D}{\NoCube}$.
To see that they are easy to refute in \LDQDRes, it is enough to construct a short refutation in \LDQRes.
By the rules of long-distance resolution , we can resolve each pair of $A_i$ and $B_i$ clauses on $x_i$ to get the $n$ clauses  $C_i = u_i \vee \bar{u}_i \vee t_i$. 
Starting with the clause $T_n$ and resolving sequentially with the $C_i$ clauses, we obtain the purely universal clause  $\bigvee_{i=1}^n ~ u_i \vee \bar{u}_i$. This can be universally reduced to  yield the empty clause, completing the refutation.

To extend the separation to systems that allow cube-learning, we slightly modify the \DoubleLongEq\ formula.
We add new existentially quantified variables at the end of the quantifier prefix, and we add to the matrix new clauses using these variables that encode \PHP. This not only preserves that the formula is false, but also makes the formula matrix unsatisfiable. Therefore, cube-learning will never be able to help in any \QCDCL\ refutation. The hardness for $\QCDCLDep{\Dord{\D}}{\D}{\NoCube}$ remains valid even after this modification, since any refutation of the modified formula must either refute the unmodified \DoubleLongEq\ formula, or refute \PHP, and  \PHP\ itself is propositionally hard for resolution.

The \LDQRes\ refutation of the modified formula remains the same as for \DoubleLongEq\ since the new clauses are completely disjoint. 
\end{proof}

\paragraph*{The $\PreRRSTrap_n$ formulas.}
	The next formula is designed to explore how adding $\Drrs$ in different ways affects the  system. It sends \QCDCL\ trails into a "trap" (of refuting the hard existential Pigeonhole Principle \PHP; see \cref{ex:illustration}) if $\Drrs$ is not used in propagation, but allows a short refutation (a contradiction on two variables) when $\Drrs$ is used. This leads to the definition of a formula inspired by the \Trapdoor\ and \DepTrap\ formulas from \cite{BeyersdorffBohm-LMCS23}(Def.~4.5)  and \cite{ChoudhuryMahajan-JAR24}(Def.~4.4) respectively.
%
		\begin{restatable}{formula}{restatePreRRSTrap}	
			\label{for:PreRRSTrap}
			The $\PreRRSTrap_n$  formula has the prefix  \\ $ \exists a ~ \forall p ~ \exists y_1, \cdots , y_{s_n} ~ \forall w ~ \forall v ~ \exists t ~ \exists x_1, \cdots , x_{s_n} ~ \forall u ~ \exists b ~ \exists q ~ \exists r ~ \exists s$, and the matrix is as given below.
			\[
			\begin{array}{lc}
				& \PHP_n^{n+1}(x_1, \cdots, x_{s_n}) \\
				\textrm{for~~} i\in[s_n]: &
				(\bar{y}_i  \vee x_i \vee u \vee b)  \  , \  (y_i \vee \bar{x}_i  \vee u  \vee b)\\
				\textrm{for~~} i\in[s_n]: &
				(y_i \vee w \vee v  \vee t \vee b) \ , \   (y_i \vee w \vee v \vee \bar{t} \vee b)  \ , \  (\bar{y}_i \vee w \vee v \vee t \vee b)  \ , \ (\bar{y}_i \vee w \vee v \vee  \bar{t} \vee b)  \\
				&  (\bar{u} \vee \bar{b})  \ , \ (v \vee \bar{b} \vee \bar{r}) \ , \ (\bar{v} \vee b \vee s) \ , \   (a \vee \bar{b}) \ , \ (\bar{a} \vee \bar{b}) \ , \ (p \vee q) \ , \ (\bar{p} \vee \bar{q}) \\ 
			\end{array}
			\] 
		\end{restatable}

		This formula has an unsatisfiable matrix (due to the presence of $\PHP$).
		
		\begin{observation}
			The variable $"w"$ is not necessary for the lower or upper bounds proved for this formula. Initialising $\PreRRSTrap|_{w=0}$ or removing the variable $"w"$ entirely affects neither the bounds nor their proofs. 
			
			However we keep it in because the $\PreRRSTrap$ formulas are defined to extend the $\Trapdoor$  formula (defined in \cite{BeyersdorffBohm-LMCS23}) which has the $"w"$ variable. Also, it shows that even if the preprocessing step (by $\Drrs$) is non-trivial and changes the formula, addition of $\Drrs$ in propagation can still make a difference.
			
		\end{observation}

\begin{restatable}{proposition}{restateclaimPreRRSTrap}
\label{clm:PreRRSTrapDrrs}
  $\Drrs(\PreRRSTrap) = \{(u,b), (v,b), (p,q)\}$.
\end{restatable}
\begin{proof}
	We look at all universal variables individually. 
	
	First consider $p$. The only other variable it shares a clause with is $q$, and $q$ does not share clauses  with any other existential variable. Therefore, $q$ is the only potential variable that can depend on $p$ in $\Drrs$ , and it indeed does so as witnessed by the path $((p,q),(\bar{q},p) )$.
	
	Next consider $w$. Since it appears in only one polarity in the matrix, therefore by definition no variable can depend on it.
	
	Third consider $v$. Consider any path starting with $\bar{v}$ and ending in $v$, and linked by existential variables (in opposing polarities) right of $v$. Such a path must begin with the clause $\bar{v}\vee b\vee s$. Since $\bar{s}$ does not even appear in the formula, the linking literal must be $b$. The next clause must contain $\bar{b}$, and also either $v$ or an existential variable right of $v$. The only such clause is $v\vee \bar{b} \vee \bar{r}$, and $r$ cannot be used to further extend the path since the positive literal $r$ does not appear in the formula. Hence the only such path $((\bar{v},b),(\bar{b},v)) $,
	and $b$ is the only existential depending on $v$ in $\Drrs$. 
	
	Finally consider $u$. $r,s$ appear in only one polarity in the axioms and $q$ is completely disjoint from any clause with $u$. So potentially only $b$ can depend on $u$. it indeed does so, because there is a path $(u,b),(\bar{b},\bar{u})$;  therefore $(u,b) \in \Drrs$. 
	
	Thus $\Drrs(\PreRRSTrap) = \{(u,b), (v,b), (p,q)\}$.
\end{proof}
Hence  \redPreRRSTrap = $\PreRRSTrap|_{w=0}$.  That is, if we preprocess using $\Drrs$, everything stays the same except that the variable $w$ "disappears".
		
		However, just preprocessing by $\Drrs$ is not enough to make this formula easy to refute. The following lemmas shows that the presence of $\Drrs$ {\em during propagation} is crucial to achieving polynomial sized refutations; its absence forces exponential size.
	
\begin{restatable}{lemma}{restatePreRRSTrapeasy}	
	\label{lem:PreRRSTrap-easyDrrsQCDCLDrrs}
	For $\CubePolicy  \in \{\NoCube, \CubeLD, \CubeD{\Drrs}\}$ the \PreRRSTrap\ formulas have polynomial size refutations in $\Drrs + \QCDCLDep{\lo}{\Drrs}{\CubePolicy}$ 
\end{restatable}
                
\begin{proof}
	By \cref{obs:nocubeImplycube}, it suffices to show polynomial size refutations in  the system $\Drrs + \QCDCLDep{\lo}{\Drrs}{\NoCube}$. 
	
	Any trail must start with a decision on $a$. A decision in either polarity propagates $\bar{b}$, and with $\Drrs$ used in propagation, further propagates $s$ since $s$ does not depend on $v$.  Next, the variable $p$ must be decided; any polarity propagates a $q$ literal. At this point $y_1$ must be decided. Since $t$ also does not depend on $v$, this decision in either polarity propagates  a $t$ literal and then a conflict. An example trail is as follows:
	$ \T = \mathbf{a}, \bar{b}, s; \mathbf{p}, \bar{q}; \mathbf{y_1} , t ,\square $. 
	The conflict reached is due to the negation of the complete tautology on $y_1$ and $t$. Thus in 4 such trails the empty clause can be learnt, completing the refutation.
	
	The refutation is described in detail as follows:
	We construct a polynomial time $\Drrs + \QCDCLDep{\lo}{\Drrs}{\NoCube}$ refutation of \PreRRSTrap. Since $\Drrs$ for the formula is $\{(u,b), (v,b), (p,q)\}$, preprocessing using $\Drrs$ reduces  the formula to 
	\[
	\begin{array}{lc}
		& \exists a \forall p \exists y_1, \cdots , y_{s_n}  \forall v \exists t  \exists x_1, \cdots , x_{s_n} \forall u \exists b \exists q \exists r \exists s\\
		& \PHP_n^{n+1}(x_1, \cdots, x_{s_n}) \\
		\textrm{for~~} i\in[s_n]: &
		\bar{y}_i  \vee x_i \vee u \vee b  \  , \  y_i \vee \bar{x}_i  \vee u  \vee b\\
		\textrm{for~~} i\in[s_n]: &
		y_i \vee v  \vee t \vee b \ , \   y_i  \vee v \vee \bar{t} \vee b \ , \ \bar{y}_i  \vee v \vee t \vee b  \ , \ \bar{y}_i \vee v \vee  \bar{t} \vee b  \\
		&  \bar{u} \vee \bar{b}  \ , \ v \vee \bar{b} \vee \bar{r} \ , \ \bar{v} \vee b \vee s \\
		& a \vee \bar{b} \ , \ \bar{a} \vee \bar{b} \ , \ p \vee q \ , \ \bar{p} \vee \bar{q} \\ 
	\end{array}
	\] 
	
	Now consider the following trail. Due to $\Drrs$ being used in propagation as well, the literal  $t$ will be propagated even before $v$ is decided, producing a conflict in the clauses involving $t$.
	
	$$T_1 = (\mathbf{a}, \bar{b}, s; \mathbf{p}, \bar{q}; \mathbf{y_1},t,\square ) $$
	where the antecedent clauses are as follows:
	\begin{eqnarray*}
		\ante(\bar{b}) &=& \bar{a} \vee \bar{b} \\
		\ante(s) &=& \bar{v} \vee b \vee s  \\
		\ante(\bar{q}) &=&  \bar{p} \vee \bar{q}  \\
		\ante(t) &=& \bar{y}_1 \vee v \vee t \vee b \\
		\ante(\square) &=& \bar{y}_1 \vee v \vee \bar{t} \vee b 
	\end{eqnarray*}
	From this trail we learn the clause $L_1 = \bar{a} \vee \bar{y}_1$. 
	
	Next construct the trail:
	$$T_2 = (\mathbf{\bar{a}}, \bar{b}, s; \mathbf{p}, \bar{q}; \mathbf{\bar{y}_1},t,\square ) $$
	where the antecedent clauses are as follows:
	\begin{eqnarray*}
		\ante(\bar{b}) &=& a \vee \bar{b} \\
		\ante(s) &=& \bar{v} \vee b \vee s  \\
		\ante(\bar{q}) &=&  \bar{p} \vee \bar{q}  \\
		\ante(t) &=& y_1 \vee v \vee t \vee b \\
		\ante(\square) &=& y_1 \vee v \vee \bar{t} \vee b 
	\end{eqnarray*}
	From this trail, we learn the clause $L_2 = a \vee y_1$.  
	
	Now consider the following third trail:
	$$T_3 = (\mathbf{a}, \bar{b},\bar{y}_1,t,\square ) $$
	with antecedent clauses as follows:
	\begin{eqnarray*}
		\ante(\bar{b}) &=& \bar{a} \vee \bar{b} \\
		\ante(\bar{y}_1) &=& L_1 =  \bar{a} \vee \bar{y}_1  \\
		\ante(t) &=& y_1 \vee v \vee t \vee b \\
		\ante(\square) &=& y_1 \vee v \vee \bar{t} \vee b 
	\end{eqnarray*}
	From here we learn the unit clause $L_3 = \bar{a}$.
	
	Finally we have the fourth trail which is fully propagated and has no decisions.	
	$$T_4 = (\bar{a}, \bar{b},y_1,t,\square ) $$
	where,
	\begin{eqnarray*}
		\ante(\bar{a}) &=& \bar{a} \\
		\ante(\bar{b}) &=& \bar{a} \vee \bar{b} \\
		\ante(y_1) &=& L_1 =  a \vee y_1  \\
		\ante(t) &=& y_1 \vee v \vee t \vee b \\
		\ante(\square) &=& y_1 \vee v \vee \bar{t} \vee b 
	\end{eqnarray*}
	From this trail we learn the empty clause $(\square)$, thus completing the refutation.
\end{proof}

\begin{restatable}{lemma}{restatePreRRSTraphard}
	\label{lem:PreRRSTrap-hardDrrsQCDCL}
For $\CubePolicy \in \{\NoCube, \CubeLD, \CubeD{\Drrs}\}$ the \PreRRSTrap\ formulas require exponential size refutations in $\Drrs + \QCDCLDep{\lo}{\Dtrv}{\CubePolicy}$  
\end{restatable}
	
\begin{proof}
	Since the \PreRRSTrap\ formulas have an unsatisfiable matrix, therefore, by \cref{obs:unsatnocubeEqcube} it suffices to show hardness in  $\Drrs + \QCDCLDep{\lo}{\Dtrv}{\NoCube}$
	
	Observe that any trail must start with a decision on $a$, propagating a $b$ literal, followed by a decision on $p$, propagating a $q$ literal:
	$ \T = \mathbf{a/\bar{a}}, \bar{b}; \mathbf{p/\bar{p}}, \bar{q}/q $. \\
	At this point in the trail, the formula matrix has reduced to 
	\[
	\begin{array}{lc}
		& \PHP_n^{n+1}(x_1, \cdots, x_{s_n}) \\
		\textrm{for~~} i\in[s_n]: &
		(\bar{y}_i  \vee x_i \vee u)   \  , \  (y_i \vee \bar{x}_i  \vee u)  \\
		\textrm{for~~} i\in[s_n]: &
		(y_i  \vee v  \vee t)  \ , \   (y_i  \vee v \vee \bar{t})  \ , \ (\bar{y}_i  \vee v \vee t)   \ , \ (\bar{y}_i \vee v \vee  \bar{t})   \\
		&  (\bar{v}  \vee s) \\
	\end{array}
	\] 
	This is effectively the matrix of the \Trapdoor\ formulas from \cite{BeyersdorffBohm-LMCS23}, with one extra clause $\bar{v}\vee s$. After this point, all decisions are made on $y$ variables propagating a corresponding $x$ variable. By the time all $y$ variables are decided, a conflict in $\PHP$ on the $x$ variables is achieved. Thus, exactly like the \Trapdoor\ formulas, refuting \PreRRSTrap\ boils down to refuting \PHP, which is known to require exponential size.
\end{proof}

\paragraph*{The \StdDepTrap\ formulas.}	
	The third  new formula \StdDepTrap\  we introduce shows the advantage of using $\Dstd$ in propagation and learning. The goal is to create clauses that can be learnt easily with $\Dstd$, but are hard to learn without it. These learned clauses enable a quick refutation in \QCDCL\ with $\Dstd$, but without them, one is stuck refuting something hard. 
%

		\begin{restatable}{formula}{restateStdDepTrap}	
			\label{formula: StdDepTrap}
			The $\StdDepTrap_n$  formula has the prefix  \\ 
			$\exists b ~ \forall w_1 ~ \exists z_1, \cdots , z_{s_n} ~ \forall w_2 ~ \exists a ~ \exists d ~ \exists c ~ \forall u ~ \exists x ~ \exists y ~ \exists p ~ \exists e_1 ~ \exists e_2 $, and the matrix is as given below.
			\[
			\begin{array}{lcc}
				& \bar{b} \vee \PHP_n^{n+1}(z_1, \cdots, z_{s_n}) & \\
				&  (y \vee p) \ , \  (y \vee \bar{p}) ,
				& (w_1 \vee e_1)  \ , \ (w_2 \vee e_2) \\
				& (b  \vee y) \ , \ (a \vee \bar{y})  \ , \  (\bar{a} \vee x) \ , \ (\bar{c} \vee u \vee \bar{x}) ,
				& (d \vee c \vee \bar{y}) \ , \ (\bar{d} \vee c \vee \bar{y})
			\end{array}
			\]
		\end{restatable}
	
		The variables $w_1,w_2,u$ are essentially ``separators'', putting existential variables into different levels. The variables $e_1,e_2,x$ are blockers, ensuring that the separators do not get reduced too early in the trail.
\begin{restatable}{proposition}{restateclaimStdDepTrapDep}
\label{clm:StdDepTrapDstd}                  
$\Dstd(\StdDepTrap) = \{(w_1,e_1), (w_2,e_2), (u,x)\}$.
\end{restatable}
\begin{proof}
	We want to show that $\Dstd(\StdDepTrap) = \{(w_1,e_1), (w_2,e_2), (u,x)\}$. For each universal variable (there are three,  $w_1,w_2,u$), we consider its dependencies. 
	
	The variable $w_1$ appears in exactly one clause, and that clause has
	just one other variable $e_1$, This variable $e_1$ is quantified to
	the right ot $w_1$ and appears in no other clause. So $e_1$, and no
	other variable, depends on $w_1$ in $\Dstd$. Similarly, $e_2$, and no
	other variable, depends on $w_2$.
	
	The variable $u$ appears only in the clause $(\bar{c} \vee u \vee \bar{x})$. Since $c$ is left of $u$ in the quantifier prefix, it does not depend on $u$ or provide a path for  other variables to depend either. Since $x$ is quantified after $u$, it does depend on $u$. However, $x$ appears only in this clause and no other clause, so it cannot provide further paths either. Therefore it is the only variable depending on $u$.

	Thus, $\Dstd(\StdDepTrap) = \{(w_1,e_1), (w_2,e_2), (u,x)\}$.
\end{proof}


When $\Dstd$ is used in propagation, we show that these formulas have short refutations.

\begin{restatable}{lemma}{restateStdDepTrapeasy}
	\label{lem:StdDepTrap-easyQCDCLDstd}
	For $\CubePolicy \in \{\NoCube, \CubeLD, \CubeD{\Drrs}\}$ the \StdDepTrap\ formulas have polynomial size refutations in  $ \QCDCLDep{\lo}{\Dstd}{\CubePolicy}$ 
\end{restatable}

\begin{proof}
	By \cref{obs:nocubeImplycube}, it suffices to show polynomial size refutations in the system $ \QCDCLDep{\lo}{\Dstd}{\NoCube}$. 
	
	For this purpose, consider the trail
	$ \T_1 =  \mathbf{\bar{b}}, y , a , x , \bar{c} , d , \square $.  \\
	Since $(u,y) \not \in \Dstd(\StdDepTrap)$,  the learnable sequence for this trail is 
	$$L_{\T_1} = \{ c \vee \bar{y}, u \vee \bar{x} \vee \bar{y}, \bar{a} \vee \bar{y}, \bar{y}, b \}.$$ 
	We choose to learn $\bar{y}$, and  then proceed to the next trail 
	$\T_2=  \bar{y} , p , \square $. 
	The learnable clauses for this trail are 
	$$L_{\T_2} = \{ y \vee \bar{p}, y, \square\}  .$$
	Thus the empty clause $\square$ is learnt, completing the refutation.
\end{proof}

However, without  $\Dstd$, the propagations force refuting the \PHP\ clauses, which is hard. 
 	
\begin{restatable}{lemma}{restateStdDepTraphard}		
	\label{lem:StdDepTrap-hardQCDCL}
	For $\CubePolicy \in \{\NoCube, \CubeLD, \CubeD{\Dtrv}\}$  the \StdDepTrap\ formulas require exponential size refutations  in  $\QCDCLDep{\lo}{\Dtrv}{\CubePolicy}$  
\end{restatable}	

\begin{proof}
	We first show that the the \StdDepTrap\ formulas require exponential size  without cube learning i.e. in $\QCDCLDep{\lo}{\Dtrv}{\NoCube}$   and then extend the argument for the case of $\QCDCLDep{\lo}{\Dtrv}{\CubeLD}$ and $\QCDCLDep{\lo}{\Dtrv}{\CubeD{\Dtrv}}$
	
	Every $\QCDCLDep{\lo}{\Dtrv}{\NoCube}$ trail must start with a decision on the variable $b$, unless and until a literal on $b$ is learnt; thenceforth a trail must begin by propagating that literal.
	
	Suppose that the $\QCDCLDep{\lo}{\Dtrv}{\NoCube}$ trail starts with a decision $b$. In such a case there are no propagations possible. Next a decision must be made on $w_1$, which may or may not propagate $e_1$. At this point, no other propagations are possible, and decisions must be made on all the $z$ variables. Since the $z$ variables are only involved in the $\PHP$ clauses, this leads to a conflict in the $\PHP$ clauses.

	Suppose that the$\QCDCLDep{\lo}{\Dtrv}{\NoCube}$  trail starts with the decision $\bar{b}$. Then there is a unique forced trail leading to a conflict, namely,
	$$ \T_1 =  \mathbf{\bar{b}}, y , a , x , \bar{c} , d , \square .$$
	The learnable sequence for this trail under regular reduction is 
	$$L_{\T_1} = \{ c \vee \bar{y}, u \vee \bar{x} \vee \bar{y}, \bar{a} \vee u \vee \bar{y}, u \vee \bar{y}, b \}.$$
	It is easy to see that learning any clause other than $b$ does not affect 
	the decisions in the trail at all.  Thus, with trails of this type, in a few stages the clause $b$ will be learnt inevitably. 
	
	To summarise, any $\QCDCLDep{\lo}{\Dtrv}{\NoCube}$ trail must start with a decision on $b$. If the decision is $b$, it leads to a conflict in the $\PHP$ clauses. If the  trails avoid decision $b$ and start with $\bar{b}$, we will eventually be forced to learn the unit clause $b$, leading to trails starting with propagating $b$, and again reaching a conflict in the $\PHP$ clauses. Therefore, any $\QCDCLDep{\lo}{\Dtrv}{\NoCube}$\ refutation of \StdDepTrap\ reduces to refuting \PHP, thus requiring exponential size.
	
	In the case of $\QCDCLDep{\lo}{\Dtrv}{\CubeLD}$ and $\QCDCLDep{\lo}{\Dtrv}{\CubeD{\Dtrv}}$, since the $\PHP$ clauses are unsatisfiable, $\bar{b}$ must be in any satisfying assignment of the matrix of \StdDepTrap. 
	Therefore for cube learning to ever play a role, the $\QCDCLDep{\lo}{\Dtrv}{\CubeLD}$ or $\QCDCLDep{\lo}{\Dtrv}{\CubeD{\Dtrv}}$ trail must start with deciding $b$ as $\bar{b}$.  However, as discussed above, trails starting with $\bar{b}$ are forced and rapidly hit a conflict; they cannot lead on to a satisfying assignment. Thus, even though the matrix of $\StdDepTrap$ is satisfiable (e.g.\ by the literals $\bar{b},y,a,x,c,u,w_1,w_2$), such assignments can never be discovered through \QCDCL\ trails.  Hence, the $\QCDCLDep{\lo}{\Dtrv}{\NoCube}$ hardness lifts as cube learning is useless, and the \StdDepTrap\ formulas continue to be hard for $\QCDCLDep{\lo}{\Dtrv}{\CubeLD}$ and $\QCDCLDep{\lo}{\Dtrv}{\CubeD{\Dtrv}}$.
\end{proof}

		\subsection{Strength Relations between the Proof Systems}
		\label{subsec:results}
                The hardness of formulas in \QCDCL\ systems with or without cube-learning and dependency schemes are collected in \cref{tab:for-sideways}; the table includes known bounds as well as those established in \cref{subsec:formulae}.
			 \begin{sidewaystable}
			\begin{adjustbox}{width=\textwidth}
				\begin{tabular}{| l | l | l | l | l |}
					\hline
					Formula 	&		Matrix         & 		Dependencies        & 			Weakest System where  Easy       & Strongest System where   Hard    \\
					\hline
					\Equality\  (\cite{BeyersdorffBH-LMCS19}-Def. 3.1)     
					& Sat                                                                            
					& \begin{tabular}[c]{@{}l@{}} $\Drrs=\emptyset$ \\ $\Dstd=\Dtrv$ \end{tabular}
					& \begin{tabular}[c]{@{}l@{}} $\QCDCLDep{\lo}{\Dtrv}{\CubeLD}$ \cite{BohmPB-AI24}-Prop 5.2 \\ $\QCDCLDep{\lo}{\Dtrv}{\CubeD{\Dtrv}}$ \cite{BohmPB-AI24}-Prop 5.2\\ 
					\QCDCL\ with $\Drrs$ \cite{ChoudhuryMahajan-JAR24}-Lemma 1\\  $\QCDCLDep{\lo}{\Dstd}{\CubeLD}$  \\  $\QCDCLDep{\lo}{\Dtrv}{\CubeD{\Dstd}}$ \end{tabular}
					& \begin{tabular}[c]{@{}l@{}} $\QCDCLDep{\lo}{\Dtrv}{\NoCube}$ \cite{BeyersdorffBohm-LMCS23}-Cor 5.8\\
						$\QCDCLDep{\lo}{\Dstd}{\NoCube}$  \end{tabular} \\
					\hline
					\TwinEquality\ (\cite{BohmPB-AI24}-Def 6.1)
					& Sat                                                                            
					& \begin{tabular}[c]{@{}l@{}}$\Drrs=\emptyset$ \\ $\Dstd=\Dtrv$\end{tabular}
					& \begin{tabular}[c]{@{}l@{}} \QCDCL\ with $\Drrs$ \cite{ChoudhuryMahajan-JAR24}-Sec 4.9 \end{tabular} 
					& \begin{tabular}[c]{@{}l@{}} $\QCDCLDep{\lo}{\Dtrv}{\CubeLD}$ \cite{BohmPB-AI24}-Prop 6.3 \end{tabular} \\
					\hline
					\Trapdoor\ (\cite{BeyersdorffBohm-LMCS23}-Def 4.5)       
					& Unsat                                                                            
					& \begin{tabular}[c]{@{}l@{}}$\Drrs=\emptyset$ \\ $\Dstd=\{(w,t)\}$\end{tabular}
					&  \begin{tabular}[c]{@{}l@{}} \QCDCL\ with $\Drrs$ \cite{ChoudhuryMahajan-JAR24}-Lemma 2 \end{tabular} 
					& \begin{tabular}[c]{@{}l@{}}$\QCDCLDep{\lo}{\Dtrv}{\CubeLD}$ (Obs.~\ref{obs:unsatnocubeEqcube},)  \\$\QCDCLDep{\lo}{\Dtrv}{\CubeD{\Dtrv}}$ (and \cite{BeyersdorffBohm-LMCS23}-Prop 4.6) \\ $\QCDCLDep{\lo}{\Dstd}{\CubeLD}$  \\ $\QCDCLDep{\lo}{\Dstd}{\CubeD{\Dstd}}$ \end{tabular} \\
					\hline
					\DepTrap\ (\cite{ChoudhuryMahajan-JAR24}-Sec 4.4)       
					& Unsat                                                                            
					& \begin{tabular}[c]{@{}l@{}}$\Drrs=\{(w,t)\}$ \cite{ChoudhuryMahajan-JAR24}\\ $\Dstd=\{(w,t)\} \cup \{(u_i,x_i)\}$\end{tabular}
					& \begin{tabular}[c]{@{}l@{}} $\QCDCLDep{\lo}{\Dtrv}{\NoCube}$  \cite{ChoudhuryMahajan-JAR24}-Lemma 3\\ $\QCDCLDep{\lo}{\Dstd}{\NoCube}$\end{tabular} 
					& \begin{tabular}[c]{@{}l@{}} \QCDCLcube\ with $\Drrs$ (\cite{ChoudhuryMahajan-JAR24}, Obs.~\ref{obs:unsatnocubeEqcube}) \end{tabular} \\
					\hline
					\TwoPHPandCT\ (\cite{ChoudhuryMahajan-JAR24}-Sec 4.5)      
					& Unsat                                                                           
					& $\Drrs = \emptyset$ \cite{ChoudhuryMahajan-JAR24}
					& \begin{tabular}[c]{@{}l@{}} $\Drrs + \QCDCLDep{\lo}{\Dtrv}{\NoCube}$  \cite{ChoudhuryMahajan-JAR24}-Lemma 5 \\$\QCDCLDep{\lo}{\Drrs}{\NoCube}$  \cite{ChoudhuryMahajan-JAR24} \end{tabular} 
					& \begin{tabular}[c]{@{}l@{}} $\QCDCLDep{\lo}{\Dtrv}{\CubeLD}$  (\cite{ChoudhuryMahajan-JAR24},Obs.~\ref{obs:unsatnocubeEqcube}) \\
						$\QCDCLDep{\lo}{\Dtrv}{\CubeD{\Dtrv}}$  (\cite{ChoudhuryMahajan-JAR24},Obs.~\ref{obs:unsatnocubeEqcube}) \\
					 $\QCDCLDep{\lo}{\Drrs}{\CubeLD}$  (\cite{ChoudhuryMahajan-JAR24},Obs.~\ref{obs:unsatnocubeEqcube}) \\
						 $\QCDCLDep{\lo}{\Drrs}{\CubeD{\Drrs}}$  (\cite{ChoudhuryMahajan-JAR24},Obs.~\ref{obs:unsatnocubeEqcube}) \end{tabular} \\
					\hline
					\PreDepTrap \ (\cite{ChoudhuryMahajan-JAR24}-Sec 4.7)     
					&  \begin{tabular}[c]{@{}l@{}}Sat\\ $\reduceDrrs$: Unsat \end{tabular}                                                                        
					& $\Drrs=\{(w,t)\}$ \cite{ChoudhuryMahajan-JAR24}
					& \begin{tabular}[c]{@{}l@{}} \QCDCL  \cite{ChoudhuryMahajan-JAR24} \\  \QCDCLDrrs \cite{ChoudhuryMahajan-JAR24} \end{tabular} 
					& \begin{tabular}[c]{@{}l@{}}\ \DrrsQCDCLcube (\cite{ChoudhuryMahajan-JAR24},Obs.~\ref{obs:unsatnocubeEqcube}) \\ \DrrsQCDCLcubeDrrs  (\cite{ChoudhuryMahajan-JAR24},Obs.~\ref{obs:unsatnocubeEqcube}) \end{tabular} \\
					\hline
					\PropDepTrap\ (\cite{ChoudhuryMahajan-JAR24}-Sec 4.8)       
					& Unsat                                                                           
					& $\Drrs=\{(w,t),(b_1,z_1), (b_2,z_2)\}$ \cite{ChoudhuryMahajan-JAR24}
					& \begin{tabular}[c]{@{}l@{}} $\QCDCLDep{\lo}{\Dtrv}{\NoCube}$   (\cite{ChoudhuryMahajan-JAR24}) \\ $\Drrs + \QCDCLDep{\lo}{\Dtrv}{\NoCube}$  (\cite{ChoudhuryMahajan-JAR24}) \end{tabular} 
					& \begin{tabular}[c]{@{}l@{}} $\QCDCLDep{\lo}{\Drrs}{\CubeLD}$  (\cite{ChoudhuryMahajan-JAR24}, Obs.~\ref{obs:unsatnocubeEqcube}) \\ 
					$\QCDCLDep{\lo}{\Drrs}{\CubeD{\Drrs}}$  (\cite{ChoudhuryMahajan-JAR24}, Obs.~\ref{obs:unsatnocubeEqcube}) \\ 
					$\Drrs + \QCDCLDep{\lo}{\Drrs}{\CubeLD}$  (\cite{ChoudhuryMahajan-JAR24}, Obs.~\ref{obs:unsatnocubeEqcube}) \\
					$\Drrs + \QCDCLDep{\lo}{\Drrs}{\CubeD{\Drrs}}$  (\cite{ChoudhuryMahajan-JAR24}, Obs.~\ref{obs:unsatnocubeEqcube})  \end{tabular} \\
					\hline
					*\DoubleLongEq      
					& Sat                                                                           
					& $\Drrs = \Dtrv$ (\cref{clm:DoubleLongEqDrrs})
					& \begin{tabular}[c]{@{}l@{}} $\QCDCLDep{\lo}{\Dtrv}{\CubeLD}$  (\cref{lem:DoubleLongEq-easyQCDCLcube} )\\  \QCDCLcube\ with $\Drrs$ (\cref{lem:DoubleLongEq-easyQCDCLcube}) \end{tabular} 
					& \begin{tabular}[c]{@{}l@{}} \QCDCL\ (\cref{lem:DoubleLongEq-hardQCDCL}) \\  \QCDCL\ with $\Drrs$ (\cref{lem:DoubleLongEq-hardQCDCL}) \end{tabular} \\
					\hline
					*\PreRRSTrap       
					& Unsat                                                                            
					&\begin{tabular}{c}$\Drrs=\{(u,b),(v,b), (p,q)\}$ \\
                                           (\cref{clm:PreRRSTrapDrrs})
					\end{tabular}
					& \begin{tabular}[c]{@{}l@{}} $\Drrs + \QCDCLDep{\lo}{\Drrs}{\NoCube}$ (\cref{lem:PreRRSTrap-easyDrrsQCDCLDrrs}) \end{tabular} 
					& \begin{tabular}[c]{@{}l@{}} $\Drrs + \QCDCLDep{\lo}{\Dtrv}{\CubeLD}$ (\cref{lem:PreRRSTrap-hardDrrsQCDCL}) \end{tabular} \\
					\hline
					*\StdDepTrap      
					& Sat                                                                            
					& \begin{tabular}{c}$\Dstd = \{(u,x),(w_1,e_1), (w_2,e_2)\}$ \\(\cref{clm:StdDepTrapDstd})
					\end{tabular}
					& \begin{tabular}[c]{@{}l@{}} $\QCDCLDep{\lo}{\Dstd}{\NoCube}$  (\cref{lem:StdDepTrap-easyQCDCLDstd}) \end{tabular} 
					& \begin{tabular}[c]{@{}l@{}}  $\QCDCLDep{\lo}{\Dtrv}{\CubeLD}$  (\cref{lem:StdDepTrap-hardQCDCL})  \end{tabular} \\
					\hline              
				\end{tabular} 
			\end{adjustbox}
			\caption{Formulas, their dependencies, and ease/hardness of refutations. \\
		"\QCDCL\ with $\Drrs$" includes  $\{\Drrs + \QCDCLDep{\lo}{\D}{\NoCube}\}$ where $\D \in \{\Dtrv,\Drrs\}$, as well as  $\QCDCLDep{\lo}{\Drrs}{\NoCube}$.
                \\
		"\QCDCLcube\ with $\Drrs$" includes 
                $\Drrs +\QCDCLDep{\lo}{\D}{\CubePolicy}$       as well as           $\QCDCLDep{\lo}{\Drrs}{\CubePolicy}$, \\ ~~~where
                $\D \in \{\Dtrv,\Drrs\}$ and $\CubePolicy\in \{\CubeLD,\CubeD{\Drrs}\}$.
}
			\label{tab:for-sideways}
			\end{sidewaystable}
                         %

Using these bounds, we now establish various relations between the newly introduced $\lo$-based  proof systems and also between them and other \QCDCL-based proof systems. \cref{fig:all-figures}  summarises these and earlier known relations in a visually clear way.
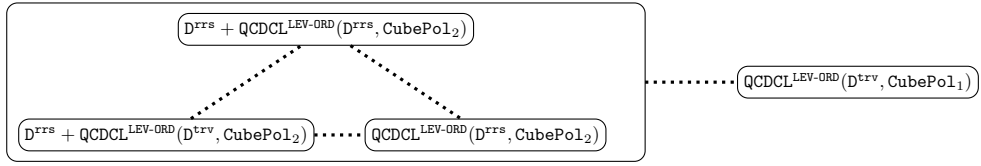
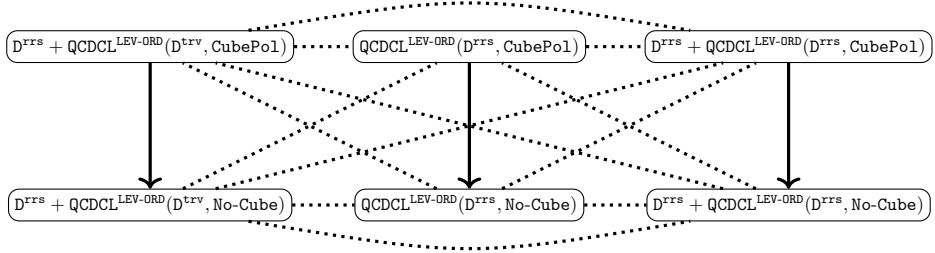
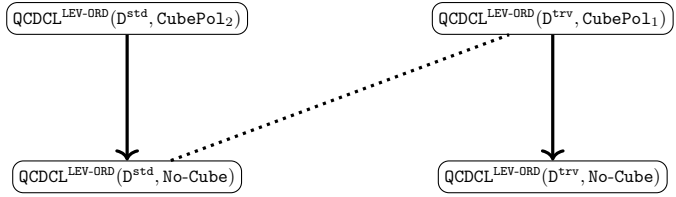
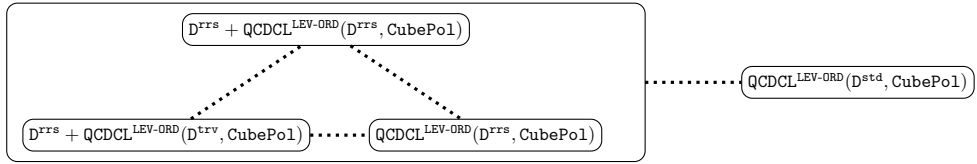
\begin{figure}
	\begin{subfigure}{\textwidth}
		\begin{tikzpicture}[scale=0.7, transform shape]
			\draw[rounded corners] (0,1) rectangle (12,4);
			\node[draw,rounded corners] (LB11) at (3,1.5) {$\Drrs+\QCDCLDep{\lo}{\Dtrv}{\CubePolicy_2}$};
			\node[draw,rounded corners] (LB12) at (9,1.5) {$\QCDCLDep{\lo}{\Drrs}{\CubePolicy_2}$};
			\node[draw,rounded corners] (LB21) at (6,3.5) {$\Drrs+\QCDCLDep{\lo}{\Drrs}{\CubePolicy_2}$};
			\node[draw,rounded corners] (RB) at (16,2.5) {$\QCDCLDep{\lo}{\Dtrv}{\CubePolicy_1}$};
			\draw[very thick, dotted] (LB11) -- (LB12);
			\draw[very thick, dotted] (LB21) -- (LB12);
			\draw[very thick, dotted] (LB11) -- (LB21);
			\draw[very thick, dotted] (12,2.5) -- (RB);
		\end{tikzpicture}
		
		\caption{For $\CubePolicy_1\in\{\NoCube,\CubeLD,\CubeD{\Dtrv}\}$; $\CubePolicy_2\in\{\CubeLD,\CubeD{\Drrs}\}$. \\ 
			From \cref{thm:QCDCLcubeDrrs-QCDCLcube,thm:QCDCLcubeDrrs-incomparable}.    }
	\end{subfigure}
	
	\par\bigskip\medskip
	
	\begin{subfigure}{\textwidth}

		\begin{tikzpicture}[scale=0.7, transform shape]
			\node[draw,rounded corners] (11) at (2,1) {$\Drrs+\QCDCLDep{\lo}{\Dtrv}{\NoCube}$};
			\node[draw,rounded corners] (12) at (8,1) {$\QCDCLDep{\lo}{\Drrs}{\NoCube}$};
			\node[draw,rounded corners] (13) at (14,1) {$\Drrs+\QCDCLDep{\lo}{\Drrs}{\NoCube}$};
			\node[draw,rounded corners] (21) at (2,4) {$\Drrs+\QCDCLDep{\lo}{\Dtrv}{\CubePolicy}$};
			\node[draw,rounded corners] (22) at (8,4) {$\QCDCLDep{\lo}{\Drrs}{\CubePolicy}$};
			\node[draw,rounded corners] (23) at (14,4) {$\Drrs+\QCDCLDep{\lo}{\Drrs}{\CubePolicy}$};
			\draw[very thick, <-] (11) -- (21);
			\draw[very thick, <-] (12) -- (22);
			\draw[very thick, <-] (13) -- (23);
			\draw[very thick, dotted] (11) -- (22);
			\draw[very thick, dotted] (11) -- (23);
			\draw[very thick, dotted] (12) -- (21);
			\draw[very thick, dotted] (12) -- (23);
			\draw[very thick, dotted] (13) -- (21);
			\draw[very thick, dotted] (13) -- (22);
			\draw[very thick, dotted] (11) -- (12);
			\draw[very thick, dotted] (12) -- (13);
			\draw[very thick, dotted] (11) .. controls (8,0) .. (13);
			\draw[very thick, dotted] (21) -- (22);
			\draw[very thick, dotted] (22) -- (23);
			\draw[very thick, dotted] (21) .. controls (8,5) .. (23);
		\end{tikzpicture}
		\caption{For $\CubePolicy \in \{\CubeLD,\CubeD{\Drrs}\}$.\\
			From \cref{thm:QCDCLcubeD-stronger,thm:QCDCLcubeDrrs-DrrsQCDCL}.  
		}
	\end{subfigure}
	
	\par\bigskip\medskip
	
	\begin{subfigure}{\textwidth}

		\begin{tikzpicture}[scale=0.7, transform shape]
			\node[draw,rounded corners] (11) at (2,1) {$\QCDCLDep{\lo}{\Dstd}{\NoCube}$};
			\node[draw,rounded corners] (12) at (10,1) {$\QCDCLDep{\lo}{\Dtrv}{\NoCube}$};
			\node[draw,rounded corners] (21) at (2,4) {$\QCDCLDep{\lo}{\Dstd}{\CubePolicy_2}$};
			\node[draw,rounded corners] (22) at (10,4) {$\QCDCLDep{\lo}{\Dtrv}{\CubePolicy_1}$};
			\draw[very thick, <-] (11) -- (21);
			\draw[very thick, <-] (12) -- (22);
			\draw[very thick, dotted] (11) -- (22);
		\end{tikzpicture}
		
		\caption{For $\CubePolicy_1\in\{\CubeLD,\CubeD{\Dtrv}\}$, $\CubePolicy_2\in\{\CubeLD,\CubeD{\Dstd}\}$.\\
			From \cref{thm:QCDCLcubeD-stronger,thm:DstdCube-incomp}. }
	\end{subfigure}
	
	\par\bigskip\medskip
	
	\begin{subfigure}{\textwidth}

		\begin{tikzpicture}[scale=0.7, transform shape]
			\draw[rounded corners] (0,1) rectangle (12,4);
			\node[draw,rounded corners] (LB11) at (3,1.5) {$\Drrs+\QCDCLDep{\lo}{\Dtrv}{\CubePolicy}$};
			\node[draw,rounded corners] (LB12) at (9,1.5) {$\QCDCLDep{\lo}{\Drrs}{\CubePolicy}$};
			\node[draw,rounded corners] (LB21) at (6,3.5) {$\Drrs+\QCDCLDep{\lo}{\Drrs}{\CubePolicy}$};
			\node[draw,rounded corners] (RB) at (16,2.5) {$\QCDCLDep{\lo}{\Dstd}{\CubePolicy}$};
			\draw[very thick, dotted] (LB11) -- (LB12);
			\draw[very thick, dotted] (LB21) -- (LB12);
			\draw[very thick, dotted] (LB11) -- (LB21);
			\draw[very thick, dotted] (12,2.5) -- (RB);
		\end{tikzpicture}

		\caption{For $\CubePolicy \in \{\NoCube,\CubeLD,\CubeD{\Dstd} \}$. From \cref{thm:DstdDrrs-incomp}.}
	\end{subfigure}
	\caption{The simulation order of various \QCDCL\ systems.\\
		$A\rightarrow B$ means $A$ simulates $B$ but $B$ does not simulate $A$.\\
		$A\cdots B$ means neither $A$ nor $B$ simulates the other.
	}
	\label{fig:all-figures}
\end{figure}

		As discussed cube-learning in \QCDCL\ proof systems is the concept of allowing "terms" to be learnt from satisfying trails, and intuitively more (optional) 
		 learning power would mean a more powerful proof system. The first theorem validates this claim . It states that for any \QCDCL\ proof system with dependency schemes, 
		 adding cube-learning yields a more powerful system than the corresponding system without it.  
		 
\begin{restatable}{theorem}{restateThmQCDCLcubeDstronger}
			\label{thm:QCDCLcubeD-stronger}
			 For $\D_1,\D_2  \in \{\Dtrv, \Dstd, \Drrs\}$ and $\CubePolicy \in \{\CubeLD, \CubeD{\D_2}\}$  the proof system $\D_1 + \QCDCLDep{\lo}{\D_2}{\CubePolicy}$ is strictly stronger than $\D_1 + \QCDCLDep{\lo}{\D_2}{\NoCube}$
\end{restatable}
		\begin{proof}
				From \cref{obs:nocubeImplycube} the systems with cube-learning are at least as strong as the corresponding systems without cube learning. The \DoubleLongEq\ formulas show that they are in fact strictly stronger; see the bounds in \cref{lem:DoubleLongEq-easyQCDCLcube,lem:DoubleLongEq-hardQCDCL}.
		\end{proof}


For the next three theorems we focus  specifically on  $\Drrs$.
The first of these shows  that irrespective of the manner in which $\Drrs$  
is used in a \QCDCL\ proof system, the resulting system is provably incomparable in strength to  the \QCDCL\ system that does not use dependency schemes. This was already known for the setting without cube-learning,  Theorem~5 in \cite{ChoudhuryMahajan-JAR24}.  We  show here that cube learning makes no difference; the systems still remain incomparable. This  highlights the fact that adding dependency schemes is not always beneficial.
              
\begin{restatable}{theorem}{restateThmQCDCLcubeDrrsQCDCLcube}
			\label{thm:QCDCLcubeDrrs-QCDCLcube}		
			Any proof systems $\PS_1,\PS_2$ are incomparable, where
                        \[\PS_1  \in  \left\{ \QCDCLDep{\lo}{\Dtrv}{\CubePolicy} \mid \CubePolicy  \in \{ \NoCube, \CubeLD, \CubeD{\Drrs}\right\}  \} \textrm{~and}\]
			\[\PS_2  \in \left\{ \D_1 + \QCDCLDep{\lo}{\D_2}{\CubePolicy} \mid 
                        \parbox{3in}{$(\D_1,\D_2)  \in \{(\Dtrv, \Drrs), (\Drrs, \Dtrv), (\Drrs, \Drrs) \}, \\  \CubePolicy \in \{\CubeLD, \CubeD{\D_2}\}$}   \right\}.\] 
\end{restatable}
  \begin{proof}
			There are eighteen incomparability claims  expressed so thirty-six separations are required! Fortunately, just two formulas establish all the desired separations. 

The bounds from \cite{ChoudhuryMahajan-JAR24} along with \cref{obs:nocubeImplycube} and  \cref{obs:unsatnocubeEqcube} imply that the \DepTrap\ formulas require exponential size refutations in every system in $\PS_2$, but have polynomial size refutations in every system in $\PS_1$.
The bounds from \cite{BeyersdorffBohm-LMCS23} and \cite{ChoudhuryMahajan-JAR24} along with \cref{obs:nocubeImplycube} imply that the \TwinEquality\ formulas require exponential size refutations in every system in $\PS_1$, but have polynomial size refutations in every system in $\PS_2$.
%
		\end{proof}

The next theorem shows that adding $\Drrs$ in different ways to \QCDCL\ yields systems incomparable in strength. (Again, this was already known for the setting without cube-learning,  Theorem~4 in \cite{ChoudhuryMahajan-JAR24}.) This highlights that adding dependency schemes in any one of  preprocessing or propagation and learning is not  inherently better than the other.

		\begin{restatable}{theorem}{restateQCDCLcubeDrrsincom}
			\label{thm:QCDCLcubeDrrs-incomparable}
			For a fixed $\CubePolicy \in \{\CubeLD, \CubeD{\Drrs}\}$, the three systems \\
                        $\QCDCLDep{\lo}{\Drrs}{\CubePolicy}$, $\Drrs+ \QCDCLDep{\lo}{\Dtrv}{\CubePolicy}$, and \\ $\Drrs + \QCDCLDep{\lo}{\Drrs}{\CubePolicy}$ are pairwise incomparable.
		\end{restatable}
\begin{proof}
	Refer to \cref{tab:for-sideways} and \cref{obs:nocubeImplycube} and \cref{obs:unsatnocubeEqcube}.
	
	It can be seen that \TwoPHPandCT\ formulas are easy to refute  when $\Drrs$ is used in preprocessing i.e. in $\Drrs + \QCDCLDep{\lo}{\Dtrv}{\CubePolicy}$ and  $\Drrs + \QCDCLDep{\lo}{\Drrs}{\CubePolicy}$, but hard otherwise i.e.  $ \QCDCLDep{\lo}{\Drrs}{\CubePolicy}$. On the other hand, the \PreDepTrap\ formulas are hard to refute if preprocessed by $\Drrs$, but easy otherwise. Together, they witness that $ \QCDCLDep{\lo}{\Drrs}{\CubePolicy}$ is incomparable with $\Drrs + \QCDCLDep{\lo}{\Dtrv}{\CubePolicy}$ and  $\Drrs + \QCDCLDep{\lo}{\Drrs}{\CubePolicy}$.
	
	Further, the \PropDepTrap\ formulas are easy to refute if the propagations and clause learning do not use $\Drrs$, but become hard if $\Drrs$ is used. This is independent of whether preprocessing is used and whether cube-learning is switched on. On the other hand, the (new) \PreRRSTrap\ formulas show that using $\Drrs$ in propagations can be advantageous. Together, these two formulas  witness   $\Drrs + \QCDCLDep{\lo}{\Dtrv}{\CubePolicy}$  and  $\Drrs + \QCDCLDep{\lo}{\Drrs}{\CubePolicy}$ are incomparable.
\end{proof}

                
Earlier, we have seen that adding cube-learning to a \QCDCL\ system with dependency always yields a stronger system, \cref{thm:QCDCLcubeD-stronger}. The following theorem shows that even adding cube-learning to a \QCDCL\ system using $\Drrs$ is incomparable to a \QCDCL\ system without cube-learning but using $\Drrs$ in a different way.
		\begin{restatable}{theorem}{restateQCDCLcubeDrrsDrrsQCDCL}
			\label{thm:QCDCLcubeDrrs-DrrsQCDCL}
				For any $\CubePolicy \in \{\CubeLD, \CubeD{\Drrs}\}$, adding $\Drrs$ to the proof system  $\QCDCLDep{\lo}{\Dtrv}{\CubePolicy}$  in one way and to $\QCDCLDep{\lo}{\Dtrv}{\NoCube}$ in any different way yields incomparable proof systems. In particular,
					\begin{enumerate}
						\item $\Drrs + \QCDCLDep{\lo}{\Dtrv}{\CubePolicy}$ is incomparable with $\QCDCLDep{\lo}{\Drrs}{\NoCube}$ and $\Drrs+\QCDCLDep{\lo}{\Drrs}{\NoCube}$.
						\item $\QCDCLDep{\lo}{\Drrs}{\CubePolicy}$ is incomparable with $\Drrs + \QCDCLDep{\lo}{\Dtrv}{\NoCube}$  and $\Drrs + \QCDCLDep{\lo}{\Drrs}{\NoCube}$.
						\item $\Drrs + \QCDCLDep{\lo}{\Drrs}{\CubePolicy}$ incomparable with $\QCDCLDep{\lo}{\Drrs}{\NoCube}$ and $\Drrs + \QCDCLDep{\lo}{\Dtrv}{\NoCube}$.
					\end{enumerate}
		\end{restatable}

\begin{proof}
	Refer to \cref{tab:for-sideways} and \cref{obs:nocubeImplycube} and \cref{obs:unsatnocubeEqcube}.
	
	The \TwoPHPandCT\ formulas are easy to refute if and only if preprocessed by $\Drrs$, irrespective of whether or not cube-learning is used and whether or not $\Drrs$ is used in propagation and learning.
	The situation is exactly reversed for the \PreDepTrap\ formulas, which are easy to refute if and only if not preprocessed by $\Drrs$. Together, these show eight of the twelve claimed incomparability relations, namely \\
	\begin{tabular}{rr}
		$\Drrs + \QCDCLDep{\lo}{\Dtrv}{\CubePolicy}$ and &$\QCDCLDep{\lo}{\Drrs}{\NoCube}$,\\
		$\QCDCLDep{\lo}{\Drrs}{\CubePolicy}$ and &$\Drrs + \QCDCLDep{\lo}{\Dtrv}{\NoCube}$,\\
		$\QCDCLDep{\lo}{\Drrs}{\CubePolicy}$ and &$\Drrs + \QCDCLDep{\lo}{\Drrs}{\NoCube}$,\\
		$\Drrs + \QCDCLDep{\lo}{\Drrs}{\CubePolicy}$ and &$\QCDCLDep{\lo}{\Drrs}{\NoCube}$.
	\end{tabular}
	
	The \PropDepTrap\ formulas are easy to refute if and only if $\Drrs$ is not used for propagation and learning, irrespective of its use in preprocessing, and irrespective of whether or not cube-learning is used.  On the other hand, the  \PreRRSTrap\ formulas are easy to refute if $\Drrs$ is used for preprocessing and in propagation and learning, but not if it is used only for preprocessing. Together, these show the remaining four claimed incomparability relations, namely
	
	\begin{tabular}[b]{rr}
		$\Drrs + \QCDCLDep{\lo}{\Dtrv}{\CubePolicy}$ and &$\Drrs+\QCDCLDep{\lo}{\Drrs}{\NoCube}$\\
		$\Drrs + \QCDCLDep{\lo}{\Drrs}{\CubePolicy}$ and &$\Drrs + \QCDCLDep{\lo}{\Dtrv}{\NoCube}$.
	\end{tabular}
\end{proof}
		
%

Now we come to $\Dstd$. First, we show that, as  for $\Drrs$, a \QCDCL\ system with $\Dstd$ is incomparable in strength to a standard \QCDCL\ system without any dependency schemes.


\begin{restatable}{theorem}{restateThmDstdCubeIncomp}
		\label{thm:DstdCube-incomp}
			For $\CubePolicy \in \{\CubeLD, \CubeD{\Dstd}\}$, the proof systems $\QCDCLDep{\lo}{\Dstd}{\NoCube}$  and $\QCDCLDep{\lo}{\Dtrv}{\CubePolicy}$ are incomparable.
\end{restatable}
		\begin{proof}
				For $\CubePolicy \in \{\CubeLD, \CubeD{\Dstd}\}$, the \Equality\ formulas have polynomial size $\QCDCLDep{\lo}{\Dtrv}{\CubePolicy}$ refutations (the refutations in \cite{BohmPB-AI24} are of this type), but require exponential size $\QCDCLDep{\lo}{\Dstd}{\NoCube}$  refutations (the lower bound from \cite{BeyersdorffBohm-LMCS23} carries over, because $\Dstd=\Dtrv$). On the other hand, the newly defined  \StdDepTrap\ formulas have polynomial size $\QCDCLDep{\lo}{\Dstd}{\NoCube}$  refutations but require exponential size $\QCDCLDep{\lo}{\Dtrv}{\CubePolicy}$  refutations (\cref{lem:StdDepTrap-easyQCDCLDstd,lem:StdDepTrap-hardQCDCL}). 
		\end{proof}

Finally, we compare the different dependency schemes $\Drrs$ and $\Dstd$. 
The mere fact that $\Drrs$ is a refinement of $\Dstd$ (more general, eliminates more dependencies) does not make it better; for that matter, $\Drrs$ is a refinement of $\Dtrv$, but using it can be a disadvantage for some formulas. Similarly, we prove below that neither of $\Drrs$ and $\Dstd$ has a  proof-theoretic advantage over the other, irrespective of the presence or absence of cube-learning.

\begin{restatable}{theorem}{restateThmDstdDrrsIncomp}	
		\label{thm:DstdDrrs-incomp}
			For any $ (\D_1,\D_2)  \in \{(\Dtrv, \Drrs), (\Drrs, \Dtrv), (\Drrs, \Drrs) \}$,\\ and $\CubePolicy \in \{\NoCube, \CubeLD, \CubeD{\Dstd} \}$, the proof systems \\
			$\PS_1 \in$ $\QCDCLDep{\lo}{\Dstd}{\CubePolicy} $ and $\PS_2 \in$  $\D_1 + \QCDCLDep{\lo}{\D_2}{\CubePolicy}$ are incomparable.
\end{restatable}
		\begin{proof}
			 Using  refutations and lower bound arguments from \cite{BeyersdorffBohm-LMCS23,ChoudhuryMahajan-JAR24} along with \cref{obs:nocubeImplycube} and \cref{obs:unsatnocubeEqcube}, the \Trapdoor\  and \DepTrap\ formulas bear witness;  the former are easy in  $\PS_2$ but hard in  $\PS_1$, whereas the situation is reversed for the latter. 
		\end{proof}

\section{Conclusions}
\label{sec:concl}

%
%
%
%

 In the context of QBF proof systems, dependency schemes are
expected to aid the process of refutation. Indeed, in the
proof systems \QRes\ and \QURes, two of the earliest
resolution-based QBF proof systems to be studied \cite{BuningKF-IC95,Gelder-CP12}, it
is known that using $\Drrs$ can shorten proofs exponentially
-- the \Equality\ formulas require exponential refutation size
in these but have polynomial-sized proofs in \QDrrsRes. It was
thus a surprise to see this advantage does not automatically
translate to QCDCL algorithms; we see that when restricted to $\lo$ decisions, 
even in the presence of cube learning, usage of dependency schemes for propagation 
and learning is not always advantageous. The hardness in these systems is primarily a consequence
of the level-ordered nature of decisions -- it can be shown that all formulas with lower bounds 
shown in \cref{sec:QCDCLcuberesults} are easy when the decision policy for QCDCL is $\Dord{\Dstd}$ or $\Dord{\Drrs}$. 
In fact, proper lower bound techniques for a decision policy $\Dord{\D}$ are unknown, and to us, this is the big open 
question arising from this work.

Even for the level-ordered policy, there are many unresolved questions.
The comparison between $\Dtrv$ and $\Drrs$ already shows that using a more refined scheme is not necessarily an advantage. The relative strengths of $\Dstd$ and $\Drrs$ (which refines $\Dstd$) is not yet clear.  How other  dependency schemes would relate  in this scenario is completely unexplored.

Since QCDCL-style reasoning is explained through long-distance clause/term resolution, it is worth highlighting three long-standing open questions about those systems. Firstly, does using dependency schemes confer any advantage with clausal long-distance; i.e.\ is the simulation of $\LDQRes$ by $\LDQDRes$ strict for  $\Dstd$ or $\Drrs$? Secondly, is the use of dependency schemes in long-distance term resolution (the system $\LDQTermRes$) sound? Thirdly, since $\Dstd$ is the scheme actually used in DepQBF, separations specific to $\Dstd$ would be quite interesting. Can we even show that $\QDstdRes$ is strictly more powerful than $\QRes$? This question has remained open since $\Dstd$ was first implemented in DepQBF over 15 years ago.

                
\bibliography{referencesOrganized}

%
\appendix

\section{Normal Dependency Schemes.}
\label{app:Dep-Schemes}

\begin{definition}
	\label{def: monodepcheme}
	A dependency scheme $\D$ is said to be a monotone dependency scheme
	if $\D(\phi[\tau]) \subseteq \D(\phi)$ for every PCNF formula $\phi$ and 
	assignment $\tau$ to subset $var(\phi)$.
\end{definition}

\begin{definition}
	\label{def: simpdepscheme}
	A dependency scheme $\D$ is said to be a simple dependency scheme if
	for every PCNF formula $\Phi = \forall X Q.\phi$,
	every $LDQ(D)$ derivation $P$ from $\Phi$, for every $u \in X$
	either $u$ or $\bar u$ do not appear in $P$
\end{definition}

\begin{definition}
	\label{def: normaldepscheme}
	A dependency scheme $\D$ is said to be a normal dependency scheme \cite{PeitlSS-JAR19}
	if: 
	\begin{itemize}
		\item $\D$ is a Monotone dependency scheme 
		\item $\D$ is a Simple dependency scheme
	\end{itemize}
\end{definition}

$\Dtrv$, $\Dstd$, $\Drrs$ are all normal dependency schemes.

\section{Rules in various QBF Proof systems}        
\label{sec:proofsystems}

For an initial PCNF and a fixed dependency scheme $\D$, rules for various proofs systems are defined as follows.

       For clauses:
       \begin{itemize}
       \item $\CAxiom$ (clause axiom rule)
              \AxiomC{$ $}
               \UnaryInfC{$A$}
               \DisplayProof where $A$ is any clause in the matrix.

       \item $\reduceD$ (clause reduction with $\D$): 
               \AxiomC{$A \vee \ell$}
               \UnaryInfC{$A$}
               \DisplayProof \\
where $\var(\ell)\in X_\forall$, and for each $x\in X_\exists \cap
         \var(A)$, $(\var(\ell),x)\not\in \D$.

       \item $\Res$ (Resolution):
               \AxiomC{$A \vee \ell$}
               \AxiomC{$B \vee \neg\ell$}
               \BinaryInfC{$A\vee B$}
               \DisplayProof \\
               where $\var(\ell)\in X_\exists$, and $A\vee B$ is not tautological.
       \item $\LDRD{\D}$ (Long-distance Resolution with $\D$):
               \AxiomC{$A \vee \ell$}
               \AxiomC{$B \vee \neg\ell$}
               \BinaryInfC{$A\vee B$}
               \DisplayProof \\
               where $\var(\ell)\in X_\exists$;  for each $x\in X_\exists$, either $x$ or $\neg x$ is not in  $A\cup B$; for each
               $u\in X_\forall$, if $u \in A$ and $\neg u \in B$, then $(u,\var(\ell))\not \in \D$. \\
               (Note that tautological clauses can be generated. In much of the literature, the symbol $u^*$ is used to denote that both $u$ and $\neg u$ are in  a clause.)
           \end{itemize}
The proof system $\QDRes$ uses the  rules $\CAxiom$, $\reduceD$, and $\Res$; the proof system $\LDQDRes$ uses the rules  $\CAxiom$, $\reduceD$, and $\LDRD{\D}$. The proof systems \LDQRes\ and  \QRes\ are the special cases of \LDQDRes\ and  \QDRes\ where $\D=\Dtrv$. A refutation of a false QBF is a derivation of the empty clause using the permitted rules. 

For cubes/terms, the situation is essentially dual, exchanging the roles of $X_\exists$ and $X_\forall$, to give the rules $\reduceD_\exists$ (term reduction with $\D$),  $\TermRes$ (Term Resolution), and $\LDTRD{\D}$ (Long-distance Term Resolution with $\D$).  Also, the
Axiom rule is modified to the term axiom rule: \AxiomC{} \UnaryInfC{$A$} \DisplayProof ~~ 
          where $A$ is a non-contradictory cube  whose literals satisfy the matrix.

The proof system $\QDTermRes$ uses the  rules  $\TAxiom$, $\reduceD_\exists$, and $\TermRes$; the proof system $\LDQDTermRes$ uses the rules $\TAxiom$,  $\reduceD_\exists$, and $\LDTRD{\D}$. The proof systems \QTermRes\ and  \LDQTermRes\ are their respective special cases where $\D=\Dtrv$. A verification of a true QBF is a derivation of the empty term using the permitted rules.

\end{document}